\documentclass[12pt,english]{amsart}
\usepackage[T1]{fontenc}
\usepackage[latin9]{inputenc}
\usepackage[letterpaper]{geometry}
\usepackage{units}
\usepackage{amsthm}
\usepackage{amstext}
\usepackage{graphicx}
\usepackage{amssymb}

\makeatletter
\numberwithin{equation}{section} 
\numberwithin{figure}{section} 
\theoremstyle{plain}
\newtheorem{thm}{Theorem}
\numberwithin{thm}{section}
\newtheorem{cor}[thm]{Corollary}
\newtheorem{lem}[thm]{Lemma}
\newtheorem{prop}[thm]{Proposition}
\newtheorem{property}[thm]{Property}
\theoremstyle{remark}
\newtheorem{remark}[thm]{Remark}
\newtheorem{definition}[thm]{Definition}

\newcommand{\Graph}{\Gamma}
\newcommand{\G}{\Gamma}

\newcommand{\Id}{\mathbb{I}}

\newcommand{\cH}{\mathcal{H}}
\newcommand{\tU}{\tilde{U}}
\newcommand{\tf}{\tilde{f}}
\newcommand{\Vertices}{\mathcal{V}}
\newcommand{\Edges}{\mathcal{E}}
\newcommand{\Leads}{\mathcal{E}^\infty}

\newcommand{\Reals}{\mathbb{R}}
\newcommand{\rmi}{{\mathrm{i}}}
\newcommand{\rme}{{\mathrm{e}}}
\newcommand{\id}{\mathbb{I}}

\renewcommand{\Re}{\mathrm{Re}}
\renewcommand{\Im}{\mathrm{Im}}

\DeclareMathOperator{\sgn}{sign}
\DeclareMathOperator{\mmod}{mod}
\DeclareMathOperator{\Range}{Range}
\DeclareMathOperator{\Ker}{Ker}

\makeatother

\usepackage{babel}

\begin{document}

\title[Nodal count on a family of quantum graphs]{Dynamics of nodal
  points and the nodal count on a family of quantum graphs}
\author{R.~Band$^{1}$} \address{$^{\text{1}}$Department of Physics of Complex Systems, The Weizmann
Institute of Science, Rehovot 76100, Israel}
\author{G.~Berkolaiko$^{2}$} \address{$^{\text{2}}$Department of Mathematics, Texas A\&M
  University, College Station, TX 77843-3368, USA}
\author{U.~Smilansky$^{1,}$$^{3}$} \address{$^{\text{3}}$Cardiff School of Mathematics and WIMCS,
Cardiff University, Senghennydd Road, Cardiff CF24 4AG, UK}

\begin{abstract}
  We investigate the properties of the zeros of the eigenfunctions on
  quantum graphs (metric graphs with a Schr\"odinger-type differential
  operator).  Using tools such as scattering approach and eigenvalue
  interlacing inequalities we derive several formulas relating the
  number of the zeros of the $n$-th eigenfunction to the spectrum of
  the graph and of some of its subgraphs.  In a special case of the
  so-called dihedral graph we prove an explicit formula that only uses
  the lengths of the edges, entirely bypassing the information about
  the graph's eigenvalues.  The results are explained from the point
  of view of the dynamics of zeros of the solutions to the scattering
  problem.
\end{abstract}

\maketitle

\section{Introduction}

Spectral properties of differential operators on graphs have recently
arisen as models for such diverse areas of research as quantum chaos,
photonic crystals, quantum wires and nanostructures.  We refer the
interested reader to the reviews \cite{Kuc_wrm02,GnuSmi_ap06} as well
as to collections of recent results
\cite{BerCarFulKuc_eds06,ExnKeaKuc_eds08}.  As a part of this research
program, the study of eigenfunctions, and in particular, their nodal
domains is an exciting and rapidly developing research direction. It
is an extension to graphs of the investigations of nodal domains on
manifolds, which started already in the 19th century by the pioneering
work of Chladni on the nodal structures of vibrating plates. Counting
nodal domains started with Sturm's oscillation theorem which states
that a vibrating string is divided into exactly $n$ nodal intervals by
the zeros of its $n$-th vibrational mode.  In an attempt to
generalize Sturm's theorem to manifolds in more than one dimension,
Courant formulated his nodal domains theorem for vibrating membranes,
which bounds the number of nodal domains of the $n$-th eigenfunction
by $n$ \cite{CourantHilbert_volume1}.  Pleijel has shown later that
Courant's bound can be realized only for finitely many eigenfunctions
\cite{Ple_cpam56}.  The study of nodal domains counts was revived
after Blum \emph{et al} have shown that nodal count statistics can be
used as a criterion for quantum chaos \cite{BGS02}.  A subsequent
paper by Bogomolny and Schmit illuminated the fascinating connection
between nodal statistics and percolation theory \cite{BS02}. A recent
paper by Nazarov and Sodin addresses the counting of nodal domains of
eigenfunctions of the Laplacian on $\mathbb{S}^{2}$ \cite{NS07}. They
prove that on average, the number of nodal domains increases linearly
with $n$, and the variance about the mean is bounded. At the same
time, it was shown that the nodal sequence - the sequence of numbers
of nodal domains ordered by the corresponding spectral parameters -
stores geometrical information about the domain \cite{GKS06}.
Moreover, there is a growing body of numerical and theoretical
evidence which shows that the nodal sequence can be used to
distinguish between isospectral manifolds \cite{GSS05,BKP07,KS08}.

As far as counting nodal domains on graphs is concerned, it has been
shown that trees behave as one-dimensional manifolds, and the analogue
of Sturm's oscillation theory applies
\cite{AlO_viniti92,PokPryObe_mz96,PokPry_rms04,Sch_wrcm06}, as long as
the eigenfunction does not vanish at any vertex.  Thus, denoting by
$\nu_{n}$ the number of nodal domains of the $n$'th eigenfunction, one
has $\nu_{n}=n$ for tree graphs.  Courant's theorem applies for the
eigenfunctions on a generic graph: $\nu_{n}\leq n$,
\cite{GnuSmiWeb_wrm04}.  It should be noted that there is a correction
due to multiplicity of the $n$-th eigenvalue and the upper bound
becomes $n+m-1$, where $m$ is the multiplicity
\cite{DavGlaLeySta_laa01}. In addition, a lower bound for the number
of nodal domains was discovered recently. It is shown in
\cite{Ber_cmp08} that the nodal domains count of the $n$-th
eigenfunction has no less than $n-\beta$ nodal domains, where $\beta$
is the number of independent cycles in the graph.  Again, this result
is valid for generic eigenfunctions, namely, the eigenfunction has no
zero entries on the vertices and belongs to a simple eigenvalue.  In a
few cases, the nodal counts of isospectral quantum graphs were shown
to be different, and thus provided further support to the conjecture
that nodal count resolves isospectrality \cite{BSS06}. A recent review
entitled {}``Nodal domains on graphs - How to count them and
why?\textquotedbl{} \cite{BanOreSmi_incol08} provides a detailed
answer to the question which appears in its title (as it was known
when the article was written). In particular, this manuscript contains
a numerically established formula for the nodal count of a specific
quantum graph, expressed in terms of the lengths of its edges. This
was the first, and to this date the only, explicit nodal count formula
for a non-trivial graph and in this manuscript we succeed in
rigorously proving it.

This leads us to focus on the study of nodal domains on quantum graphs
from a new point of view. Namely, we shall show that one can count the
number of nodal domains by using scattering data obtained by attaching
semi-infinite leads to the graph. Scattering on graphs was proposed as
a paradigm for chaotic scattering in \cite{KotSmi_prl00,KotSmi_jpa03}
with new applications and further developments in the field described
in \cite{DP10,DEL10,EL10}.  The work presented here is based on the
concepts and ideas developed in these studies.

The paper is organized in the following way. The current section
provides the necessary definitions and background from the theory of
quantum graphs. The conversion of finite graphs to scattering systems
by adding leads will be discussed in the next section and the
expression for the scattering matrix will be derived and studied in
detail. The connection of the scattering data with nodal domains and
the counting methods it yields will be presented in section
\ref{sec:applications_to_nodal_count}. Section
\ref{sec:disjoint_cycles} applies the above counting methods in order
to derive a formula for the nodal count of graphs with disjoint
cycles. This formula relates the nodal count to the spectra of the
graph and some of its subgraphs. Thus, information about the
eigenfunctions is exclusively obtained from the eigenvalue
spectrum. The last section relates the different ways of counting and
discusses possible future developments.

\subsection{Quantum graphs}
\label{sec:intro_quantum_graphs}

In this section we describe the quantum graph which is a metric graph
with a Shr\"odinger-type self-adjoint operator defined on it.  Let
$\Graph = (\Vertices, \Edges)$ be a connected graph with vertices
$\Vertices=\{v_j\}$ and edges $\Edges=\{e_j\}$.  The sets $\Vertices$
and $\Edges$ are required to be finite.

We are interested in metric graphs, i.e. the edges of $\Graph$ are
$1$-dimensional segments with finite positive lengths
$\{L_e\}_{e\in\Edges}$.  On the edge $e=(u,v)$ we use two coordinates,
$x_{e,v}$ and $x_{e,u}$.  The coordinate $x_{e,v}$ measures the
distance along the edge starting from the vertex $v$; $x_{e,u}$ is
defined similarly.  The two coordinates are connected by $x_{e,v} +
x_{e,u} = L_e$.  Sometimes, when the precise nature of the coordinate
is unimportant, we will simply write $x_e$ or even $x$.

A metric graph becomes quantum after being equipped with an additional
structure: assignment of a self-adjoint differential operator.  This operator will
be often called the {\bf Hamiltonian}.  In this paper we study the
zeros of the eigenfunctions of the negative second
derivative operator ($x$ is the coordinate along an edge)
\begin{equation}
  H\ :\ f(x) \mapsto -\frac{d^2f}{dx^2}.  \label{E:lap_op}
\end{equation}
or the more general Schr\"{o}dinger operator
\begin{equation}
  H\ :\ f(x) \mapsto -\frac{d^2f}{dx^2}+V(x)f(x),  \label{E:schrod}
\end{equation}
where $V(x)$ is a \emph{potential}.  Note that the value of a function
or the second derivative of a function at a point \emph{on the edge}
is well-defined, thus it is not important which coordinate, $x_{e,v}$
or $x_{e,u}$ is used.  This is in contrast to the first derivative
which changes sign according to the direction of the chosen
coordinate.

To complete the definition of the operator we need to specify its
domain.

\begin{definition}
  \label{D:spaces}
  We denote by $\widetilde{H}^2(\G)$ the space
  \begin{equation*}\label{E:tilde_space}
    \widetilde{H}^2(\G):=\bigoplus_{e\in\Edges}H^2(e),
  \end{equation*}
  which consists of the functions $f$ on $\G$ that on each edge $e$
  belong to the Sobolev space $H^2(e)$.  The restriction of $f$ to the
  edge $e$ is denoted by $f_e$.  The norm in the space
  $\widetilde{H}^2(\G)$ is
  \begin{equation*}\label{E:tilde_space_sum}
    \|f\|_{\widetilde{H}^2(\G)} := \sum\limits_{e\in\Edges}\|f_e\|^2_{H^2(e)}.
  \end{equation*}
\end{definition}

Note that in the definition of $\widetilde{H}^2(\G)$ the smoothness is
enforced along edges only, without any junction conditions at the
vertices at all.  However, the standard Sobolev trace theorem (e.g.,
\cite{EdmundsEvans_spectral}) implies that each function $f_e$ and its
first derivative have well-defined values at the endpoints of the edge
$e$.  Since the direction is important for the first derivative, we
will henceforth adopt the convention that, at an end-vertex of an edge
$e$, the derivative is calculate \emph{into} the edge and away from
the vertex.  That is the coordinate $x$ is chosen so that the vertex
corresponds to $x=0$.

To complete the definition of the operator we need to specify its
domain.  All conditions that lead to the operator \eqref{E:lap_op}
being self-adjoint have been classified in
\cite{KosSch_jpa99,Har_jpa00,Kuc_wrm04}.  We will only be interested
in the so-called extended $\delta$-type conditions, since they are the
only conditions that guarantee continuity of the eigenfunctions,
something that is essential if one wants to study changes of sign of
the said eigenfunctions.

\begin{definition}
  \label{def:domain}
  The domain $\cH$ of the operator \eqref{E:schrod}
  consists of the functions $f \in \widetilde{H}^2(\G)$ such that
  \begin{enumerate}
  \item $f$ is continuous on every vertex:
    \begin{equation}
      \label{eq:delta_cont}
      f_{e_1}(v) = f_{e_2}(v),
    \end{equation}
    for every vertex $v$ and edges $e_1$ and $e_2$ that have $v$ as an
    endpoint.
  \item the derivatives of $f$ at each vertex $v$ satisfy
    \begin{equation}
      \label{eq:delta_deriv}
      \sum_{e \in \Edges_v}\frac{df}{dx_e}(v)=\alpha _v f(v), \qquad
      \alpha_v \in \Reals,
    \end{equation}
    where $\Edges_v$ is the set of edges incident to $v$.
  \end{enumerate}
\end{definition}

Sometimes the condition \eqref{eq:delta_deriv} is written in a more
robust form
\begin{equation}
  \label{eq:delta_deriv_tan}
  \cos(\gamma_v) \sum_{e \in \Edges_v}\frac{df}{dx_e}(v) =
  \sin(\gamma_v) f(v),
\end{equation}
which is also meaningful for infinite values of $\alpha_v =
\tan(\gamma_v)$.  Henceforth we will understand $\alpha_v=\infty$ as
the Dirichlet condition $f(v)=0$.  The case $\alpha_v=0$ is often
referred to as the Neumann-Kirchhoff condition.

Finally, we will assume that the potential $V(x)$ is bounded and
piecewise continuous.  To summarize our discussion, the operator
\eqref{E:schrod} with the domain $\cH$ is self-adjoint for any choice
of real $\alpha_v$.  Since we only consider compact graphs, the
spectrum is real, discrete and with no accumulation points.  We will
slightly abuse notation and denote by $\sigma(\Gamma)$ the spectrum of
an operator $H$ defined on the graph $\Gamma$.  It will be clear from
the context which operator $H$ we mean and what are the vertex
conditions.

The eigenvalues $\lambda \in \sigma(\Gamma)$ satisfy
the equation
\begin{equation}
  \label{eq:eig_eq}
  -\frac{d^2f}{dx^2}+V(x)f(x) = \lambda f(x).
\end{equation}
It can be shown that under the conditions specified above the operator
$H$ is bounded from below \cite{Kuc_wrm04}.  Thus we can number the
eigenvalues in the ascending order, starting with $1$.  Sometimes we
abuse the notation further and also call $k$, such that $\lambda=k^2$,
an eigenvalue of the graph $\Gamma$.  This also should lead to no
confusion since, with the conditions $\Re(k) \geq 0$, $\Im(k) \geq 0$,
the relation between $k$ and $\lambda$ is bijective.

\subsection{Nodal count}
\label{sec:intro_nodal_count}

The main purpose of this article is to investigate the number of zeros
and the number of nodal domains of the eigenfunctions of a connected
quantum graph.  We aim to give formulas linking these quantities to
the geometry of the graphs and to the eigenvalues of the graph and its
subgraphs, but avoiding any reference to the values of the
eigenfunctions themselves.

The number of internal zeros or \emph{nodal points} of the function
$f$ will be denoted by $\mu(f)$.  We will use the shorthand $\mu_n$ to
denote $\mu(f_n)$ where $f_n$ is the $n$-th eigenfunction of the graph
in question.  The sequence $\{\mu_n\}$ will be called the \emph{nodal
  point count} sequence.  A \emph{positive (negative) domain\/} with
respect to $f$ is a maximal connected subset in $\G$ where $f$ is
positive (correspondingly, negative).  The total number of positive
and negative domains will be called the \emph{nodal domain count} of
$f$ and denoted by $\nu(f)$.  Similarly to $\mu_n$, we use $\nu_n$ as
a short-hand for $\nu(f_n)$ and refer to $\{\nu_n\}$ as the
\emph{nodal domain count sequence}.

The two quantities $\mu_n$ and $\nu_n$ are closely related, although,
due to the graph topology, the relationship is more complex than on a
line, where $\nu=\mu+1$.  Namely, one can easily establish the bound
\begin{equation}\label{E:mu_nu}
  \mu - \beta_\G + 1 \leq \nu \leq \mu + 1,
\end{equation}
where $\beta_\G$ is the cyclomatic number of $\Graph$.  The
cyclomatic number can be computed as
\begin{equation}
  \label{eq:cyclom}
  \beta_\Gamma = |\Edges| - |\Vertices| + 1.
\end{equation}
The cyclomatic number has several related interpretations: it counts
the number of independent cycles in the graph (hence the name) and
therefore it is the first Betti number of $\Graph$ (hence the notation
$\beta$).  It also counts the minimal number of edges that need to be
removed from $\Gamma$ to turn it into a tree.  Correspondingly,
$\beta_\Gamma=0$ if and only if $\Gamma$ is a tree.

There is another simple but useful observation relating the cycles on
the graph and the number of zeros: if the eigenfunction of the graph
does not vanish on the vertices of the graph, the number of zeros on
any cycle of the graph is even.  Indeed, an eigenfunction of a second
order operator can only have simple zeros, thus at every zero $f$
changes sign.  On a cycle there must be an even number of sign
changes.

As mentioned earlier, we will be interested in the number of zeros and
nodal domains of the eigenfunctions of operators~\eqref{E:lap_op} and
\eqref{E:schrod} on graphs.  According to the well known ODE theorem
by Sturm \cite{Stu_jmpa36,Stu_jmpa36a,Hin_incol05}, the zeros of the
$n$-th eigenfunction of the operator of type \eqref{E:schrod} on an
interval divide the interval into $n$ nodal domains.  By contrast, in
the corresponding question in $\Reals^2$ only an upper bound is
possible, given by the Courant's nodal line theorem
\cite{CourantHilbert_volume1}, $\nu_n \leq n$.  In a series of papers
\cite{AlO_viniti92,PokPryObe_mz96,GnuSmiWeb_wrm04,Sch_wrcm06,Ber_cmp08},
it was established that a generic eigenfunction of the quantum graph
satisfies both an upper and a lower bound.  Namely, let $\lambda_n$ be
a simple eigenvalue of $-\frac{d^2}{dx^2} + V(x)$ on a graph $\Gamma$
and its eigenfunction $f_n$ be non-zero at all vertices of
$\Gamma$. Then the number of the nodal domains of $f_n$ satisfies
\begin{equation}
 \label{eq:nodal_domains_bound}
  n - \beta_\G  \leq \nu_n \leq n.
\end{equation}
Similarly, for the number of zeros we have
\begin{equation}
 \label{eq:nodal_zeros_bound}
  n - 1 \leq \mu_n \leq n - 1 +\beta_\G.
\end{equation}
Note that the upper bound in \eqref{eq:nodal_zeros_bound} follows from
the upper bound in \eqref{eq:nodal_domains_bound} and inequality
\eqref{E:mu_nu}.  The lower bound in \eqref{eq:nodal_zeros_bound}
requires an independent proof which is given in \cite{BerKuc_prep10}.  An
interesting feature of quantum graphs is that, unlike the $\Reals^d$ case,
the upper bound $\nu_n \leq n$ is in general not valid for degenerate
eigenvalues.

In the present paper we combine these a priori bounds with scattering
properties of a certain family of graphs to derive formulas for the
nodal counts $\mu_n$ and $\nu_n$.

\subsection{Quantum evolution map}
\label{sec:scat_mat}

When the potential $V(x)$ is equal to zero, the eigenvalue equation
\begin{equation}
  \label{eq:eig_eq_V0}
  -\frac{d^2f}{dx^2} = k^2 f(x),
\end{equation}
has, on each edge, a solution that is a linear combination of the two
exponents $\rme^{\pm\rmi k x}$ if $k\neq 0$.  We will write it in the
form
\begin{equation}
  \label{eq:sol_on_edge_a}
  f_e(x_{e,v}) = a_{e,v}^{\,in} \exp(-\rmi k x_{e,v}) + a_{e,v}^{\,out} \exp(\rmi k x_{e,v}),
\end{equation}
where the variable $x_{e,v}$ measures the distance from the vertex
$v$ of the edge $e$.  The coefficient $a_{e,v}^{\,in}$ is the incoming
amplitude on the edge $e$ (with respect to the vertex $v$) and
$a_{e,v}^{\,out}$ is correspondingly the outgoing amplitude.
However the same function can be expressed using the coordinate
$x_{e,u}$ as
\begin{equation}
  \label{eq:sol_on_edge_other}
  f_e(x_{e,u}) = a_{e,u}^{\,in} \exp(-\rmi k x_{e,u}) + a_{e,u}^{\,out} \exp(\rmi k x_{e,u}).
\end{equation}
Since these two expressions should define the same function and since
the two coordinates are connected, through the identity $x_{e,v} +
x_{e,u} = L_e$, we arrive to the following relations
\begin{equation}
  \label{eq:a_connection}
  a_{e,v}^{\,in} = \rme^{\rmi k L_e} a_{e,u}^{\,out}
  \qquad
  a_{e,u}^{\,in} = \rme^{\rmi k L_e} a_{e,v}^{\,out}
\end{equation}

Fixing a vertex $v$ of degree $d_v$ and using \eqref{eq:sol_on_edge_a}
to describe the solution on the edges $\Edges_v$ adjacent to $v$ we
obtain from \eqref{eq:delta_cont} and \eqref{eq:delta_deriv} $d_v$
equations on the $2d_v$ variables $a_{e,v}^{\,out}$ and
$a_{e,v}^{\,in}$.  These equations can be rearranged as
\begin{equation}
  \label{eq:vert_scat}
  \vec{a}_v^{\,out} = \sigma^{(v)}(k) \vec{a}_v^{\,in},
\end{equation}
where $\vec{a}_v^{\,out}$ and $\vec{a}_v^{\,in}$ are the vectors of
the corresponding coefficients and $\sigma^{(v)}(k)$ is a $d_v\times
d_v$ unitary matrix.  The matrix $\sigma^{(v)}(k)$ is called the
vertex-scattering matrix, it depends on $k$ for values of $\alpha_v$
other than $0$ or $\infty$ and its entries have been calculated in
\cite{KotSmi_ap99}.

Collect all coefficients $a_{e,v}^{\,in}$ into a vector $\vec{a}$ of size
$2|\Edges|$ and define the matrix $J$ acting on $\vec{a}$ by requiring
that it swaps around $a_{e,v}^{\,in}$ and $a_{e,u}^{\,in}$ for all $e=(u,v)$.
Then, collecting equations \eqref{eq:vert_scat} into one system and
using connection~\eqref{eq:a_connection} and the matrix $J$ to rewrite
everything in terms of $\vec{a}$ we have
\begin{equation*}
    J \rme^{-\rmi k L} \vec{a} = \Sigma(k) \vec{a}.
\end{equation*}
Here all matrices have the dimension equal to double the number of
edges, $2|\Edges|$.  The matrix $L$ is the diagonal matrix of edge
lengths, each length appearing twice and $\Sigma$ is the
block-diagonalizable matrix with individual $\sigma^{(v)}$ as blocks,
namely
\begin{equation*}
  \Sigma_{(e_1,v_1), (e_2,v_2)} = \delta_{v_1,v_2}
  \sigma^{(v_1)}_{e_1, e_2}.
\end{equation*}
Noting that $J^{-1}=J$, the condition on $\vec{a}$ can be rewritten as
\begin{equation}
  \label{eq:bond_scat}
  \vec{a} = \rme^{\rmi k L} J \Sigma(k) \vec{a},
\end{equation}
The unitary matrix $U(k):= \rme^{\rmi k L} J \Sigma(k)$ is variously
called the \emph{bond scattering matrix} \cite{KotSmi_ap99} or the
\emph{quantum evolution map} \cite{GnuSmi_ap06}.  The matrix $J
\Sigma(k)$ describes the scattering of the waves on the vertices of
the graph and $\rme^{\rmi k L}$ gives the phase shift acquired by the
waves while traveling along the edges.  The quantum evolution map can
be used to compute the non-zero eigenvalues of the graph through the
equation
\begin{equation}
  \label{eq:secular_cond}
  \det\left(\Id - U(k) \right) = 0.
\end{equation}
We stress that $U(k)$ is not a scattering matrix in the conventional
sense, since the graph is not open.  Turning graph into a scattering
system is the subject of the next section.

\section{Attaching infinite leads to the graph}
\label{sec:scattering_matrix}

A quantum graph may be turned into a scattering system by attaching
any number of infinite leads to some or all of its vertices. This idea
was already discussed in
\cite{KotSmi_ap99,KotSmi_jpa03,KotScha_wrm04}.  We repeat it here and
further investigate the analytic and spectral properties of the
graph's scattering matrix, that would enable the connection to the
nodal count.

Let $\Gamma=\left(\Vertices,\Edges\right)$ be a quantum graph. We
choose some $M\leq\left|\Vertices\right|$ out of its vertices and
attach to each of them an infinite lead. We call these $M$ vertices,
the marked vertices, and supply them with the same vertex conditions
as they had in $\Gamma$. Namely, each marked vertex $v$ retains its
$\delta$-type condition with the same parameter $\alpha_{v}$ (recall
\eqref{eq:delta_deriv}). We denote the extended graph that contains
the leads by $\tilde{\Gamma}$ and investigate its generalized
eigenfunctions.

The solution of the eigenvalue equation, \eqref{eq:eig_eq_V0}, on a
lead $l$ which is attached to the vertex $v$, can be written in the
form
\begin{equation}
  f_{l}(x_{l,v}) = c_{l,v}^{\,in}\exp(-\rmi kx_{l,v})
  + c_{l,v}^{\,out}\exp(\rmi kx_{l,v}).
\label{eq:sol_on_lead}
\end{equation}
The variable $x_{l,v}\in [0,\infty)$ measures the distance from the
vertex $v$ along the lead $l$ and the coefficients $c_{l,v}^{\,in},\,
c_{l,v}^{\,out}$ are the incoming and outgoing amplitudes on the lead
$l$ (compare with \eqref{eq:sol_on_edge_a}).  We use the notation
$\vec{c}^{\,out}$, $\vec{c}^{\,in}$ for the vectors of the
corresponding coefficients and follow the derivation that led to
\eqref{eq:bond_scat} in order to obtain
\begin{equation}
  \left(
    \begin{array}{c}
      \vec{c}^{\,out}\\
      \vec{a}
    \end{array}
  \right)
  =\rme^{\rmi kL_o}J_o\Sigma(k) \left(
    \begin{array}{c}
      \vec{c}^{\,in}\\
      \vec{a}
    \end{array}
  \right).
  \label{eq:full_graph_scatter0}
\end{equation}
All the matrices above are square matrices of dimension
$2\left|\Edges\right|+M$.  There are two differences from
equation~\eqref{eq:bond_scat}.  First, in the matrix $L_o$ each lead
is represented by a single zero on the diagonal, in contrast to the
positive lengths of the graph edges, appearing twice each.  The matrix
$J_o$ swaps around the coefficients $a$ corresponding to opposite
directions on internal edges, but acts as an identity on the leads.
These differences arise because for an infinite lead we do not have
two representations \eqref{eq:sol_on_edge_a} and
\eqref{eq:sol_on_edge_other} and therefore no connection formulas
\eqref{eq:a_connection} allowing to eliminate outgoing coefficients.
Writing the matrix $\rme^{\rmi kL}J\Sigma(k)$ in blocks corresponding
to the edge coefficients and lead coefficients results in
\begin{equation}
  \left(\begin{array}{c}
      \vec{c}^{\,out}\\
      \vec{a}
    \end{array}\right)
  =
  \left(\begin{array}{cc}
      R(k) & T_o(k)\\
      T_i(k) & \tilde{U}(k)
    \end{array}\right)
  \left(\begin{array}{c}
      \vec{c}^{\,in}\\
      \vec{a}
    \end{array}\right),
  \label{eq:full_graph_scatter}
\end{equation}
where the dimensions of the matrices $R$, $T_o$, $T_i$ and $\tU$ are
$M\times M$, $M\times2|\Edges|$, $2|\Edges|\times M$ and
$2|\Edges|\times2|\Edges|$ correspondingly.  We stress that the matrix
$\tilde{U}(k)$ describes the evolution of the waves inside the compact
graph and has eigenvalues that can now lie \emph{inside} the unit
circle due to the ``leaking'' of the waves into the leads.

Equation \eqref{eq:full_graph_scatter} can be used to define a unitary
scattering matrix $S$ such that $\vec{c}^{\,out}=S\,\vec{c}^{\,in}$, as
described in the following theorem.

\begin{thm}
  \label{thm:scat_mat}
  Let
  \begin{equation}
    \label{eq:full_graph_scatter_Q}
    \left(\begin{array}{c}
        \vec{c}^{\,out}\\
        \vec{a}
      \end{array}\right)
    = Q
    \left(\begin{array}{c}
        \vec{c}^{\,in}\\
        \vec{a}
      \end{array}\right),
    \qquad \mbox{where} \qquad
    Q=\left(\begin{array}{cc}
        R & T_{o}\\
        T_{i} & \tU
      \end{array}\right)
  \end{equation}
  is a unitary matrix with the blocks $R$, $T_o$, $T_i$ and $\tU$ of
  sizes $M\times M$, $M\times2|\Edges|$, $2|\Edges|\times M$ and
  $2|\Edges|\times2|\Edges|$ correspondingly.  For every choice of
  $\vec{c}^{\,in}\in\mathbb{C}^{M}$, consider
  relation~\eqref{eq:full_graph_scatter_Q} as a set of linear
  equations in the variables $\vec{c}^{\,out}\in\mathbb{C}^{M}$ and
  $\vec{a}\in\mathbb{C}^{2|\Edges|}$.  Then
  \begin{enumerate}
  \item \label{enu:scat-1} There exists at least one $2|\Edges|\times
    M$ matrix $C$ such that
    \begin{equation}
      \left(\id-\tU\right)C = T_{i},
      \label{eq:C-definition}
    \end{equation}
  \item \label{enu:scat-2} Let
    \begin{equation}
      \label{eq:S-definition}
      S = R + T_{o} C.
    \end{equation}
    Then $S$ is a unitary matrix independent of the particular choice of
    $C$ in equation~\eqref{eq:C-definition}.
  \item \label{enu:scat-3} The solutions of \eqref{eq:full_graph_scatter} are given by
    \begin{align}
      \label{eq:c_out}
      &\vec{c}^{\,out} = S \vec{c}^{\,in}, \\
      &\vec{a} \in \{ C \vec{c}^{\,in} + \Ker(\id-\tU)\}.
    \end{align}
    In particular, $\vec{c}^{\,out}$ is defined uniquely by $\vec{c}^{\,in}$.
  \end{enumerate}
\end{thm}

The proof of the theorem distinguishes between the case of a trivial
$\Ker(\id-\tU)$ and the case of singular $\id-\tU$. The
following lemma makes the treatment of the latter case easier.

\begin{lem}
  \label{lem:ker}
  Let $Q$ be as in theorem \ref{thm:scat_mat}.
  Then the following hold:
  \begin{align}
    \label{eq:ker-1}
    &\Range T_{i} \subseteq \Range\,(\id-\tU), \\
    \label{eq:ker-2}
    &\Ker\,(\id-\tU) \subseteq \Ker T_o
  \end{align}
\end{lem}

\begin{proof}
  Since in a
  finite-dimensional space $\Range A = (\Ker A^*)^\perp$,
  equation \eqref{eq:ker-1} is equivalent to
  \begin{equation*}
    \left(\Ker T_{i}^*\right)^\perp \subseteq \left(\Ker\,(\id-\tU^*)\right)^\perp,
  \end{equation*}
  which is in turn equivalent to
  \begin{equation*}
    \Ker\,(\id-\tU^*) \subseteq \Ker T_{i}^*.
  \end{equation*}
  Let $\vec{v}\in\Ker(\id-\tU^{*})$.  Using the unitarity of $Q^{*}$
  we get
  \begin{equation}
    \left\Vert \left(\begin{array}{c}
          0\\
          \vec{v}
        \end{array}\right)\right\Vert
    = \left\Vert \left(\begin{array}{cc}
          R^{*} & T_{i}^{*}\\
          T_{o}^{*} & \tU^{*}
        \end{array}\right)
      \left(\begin{array}{c}
          0\\
          \vec{v}
        \end{array}\right)\right\Vert
    =\left\Vert \left(\begin{array}{c}
          T_{i}^{*}\vec{v}\\
          \tU^{*}\vec{v}\end{array}\right)\right\Vert
    =\left\Vert \left(\begin{array}{c}
          T_{i}^{*}\vec{v}\\
          \vec{v}\end{array}\right)\right\Vert.
    \label{eq:existence_of_sol}
  \end{equation}
  Equating the left-hand side to the right-hand side of the equation
  above we get $T_{i}^{*}\vec{v}=0$.  Equation \eqref{eq:ker-2} is
  proved in a similar manner by replacing $Q^{*}$ with $Q$ in the
  above.
\end{proof}

\begin{proof}[Proof of theorem \ref{thm:scat_mat}]
  \emph{Case 1:} $\det(\id-\tU)\neq0$.

  To show part \eqref{enu:scat-1} we simply set $C=(\id-\tU)^{-1}\,T_{i}$.
  Furthermore, equation \eqref{eq:full_graph_scatter} has a unique solution, given by
  \begin{align}
    \vec{a} & =\left(\id-\tU\right)^{-1}T_{i}\vec{c}^{\, in}\label{eq:avec_in_terms_of_cinvec}\\
    \vec{c}^{\, out} & =
    \left(R+T_{o}\left(\id-\tU\right)^{-1}T_{i}\right)\vec{c}^{\, in},\label{eq:coutvec_in_terms_of_cinvec}
  \end{align}
  which proves part \eqref{enu:scat-3}.

  The unique definition of $C$ guarantees the uniqueness of
  $S=R+T_{o}C$.  To finish the proof of part \eqref{enu:scat-2} we use
  the unitarity of $Q$, which provides the identities
  \begin{align}
    R^{*}R+T_{i}^{*}T_{i} & = T_{o}^{*}T_{o}+\tU^{*}\tU=\id\nonumber \\
    R^{*}T_{o}+T_{i}^{*}\tU & =
    T_{o}^{*}R+\tU^{*}T_{i}=\mathbf{0}.
    \label{eq:unitarity_of_Q}
  \end{align}
  From here we get
  \begin{align*}
    S^{*}S &= R^{*}R+R^{*}T_{o}C+C^{*}T_{o}^{*}R+C^{*}T_{o}^{*}T_{o}C \\
    &= \id-T_{i}^{*}T_{i}-T_{i}^{*}\tU
    C-C^{*}\tU^{*}T_{i}+C^{*}\left(\id-\tU^{*}\tU\right)C.
 \end{align*}
  Expanding, factorizing and using the definition of $C$ in the form
  $T_i + \tU C =C$, we arrive to
  \begin{align*}
    S^{*}S &= \id + C^*C - (T_i^* + C^*\tU^*) (T_i + \tU C) = \id
  \end{align*}

  \emph{Case 2:} $\det(\id-\tU)=0$.

  Existence of a solution $C$ to the equation $(\id-\tU)C=T_{i}$ is
  guaranteed by equation~\eqref{eq:ker-1} of Lemma~\ref{lem:ker}.

  The columns of $C$ are defined up to addition of arbitrary vectors
  from $\Ker\,(\id-\tU)$.  However, Lemma~\ref{lem:ker},
  equation~\eqref{eq:ker-2} implies that these vectors are in the
  null-space of $T_o$, therefore the product $T_oC$ has unique value
  independent of the particular choice of the solution $C$.  The proof
  of the unitarity of $S$ has already been given in case 1 and did not
  rely on the invertibility of $\id-\tU$.  This proves part
  \eqref{enu:scat-2}.

  The last $2|\Edges|$ equations of \eqref{eq:full_graph_scatter} are
  $(\id-\tU)\vec{a}=T_{i}\vec{c}^{\,in}$.
  From~\eqref{eq:C-definition}, all solutions $\vec{a}$ of this
  equation are given by $C\vec{c}^{\,in} + \Ker(\id-\tU)$.  On the
  other hand, the first $M$ equations of \eqref{eq:full_graph_scatter}
  are $\vec{c}^{\,out} = R \vec{c}^{\,in} + T_o \vec{a}$ and
  substituting the already obtained expression for $\vec{a}$ and using
  \eqref{eq:ker-2} we finally arrive to $\vec{c}^{\,out} =
  (R+T_oC) \vec{c}^{\,in}$.  This finished the proof of the theorem.
\end{proof}

We would like to study the unitary scattering matrices, $S(k)$ as a
one-parameter family in $k\in\Reals$.  The matrix $Q(k)$ is a meromorphic
function of $k$ in the entire complex plane \cite{KosSch_jpa99}.  For
all $k$ values which satisfy $\det(\id-\tU)\neq0$, $S(k)$ is given
explicitly by
\begin{equation}
  \label{eq:S_explicit}
  S(k)=R+T_{o}\left(\id-\tU\right)^{-1} T_{i},
\end{equation}
and $S(k)$ is therefore also a meromorphic function in $k$ at these
values.  The significance of the values of $k$ for which
$\det(\id-\tU)=0$ is explained in the following lemma.

\begin{lem}
  \label{lem:special_k_values}
  Let $\Gamma^{*}$ be the quantum graph obtained from the original
  compact quantum graph, $\Gamma$, by imposing the condition $f(v)=0$ at all of its $M$ marked vertices, \emph{in addition} to the conditions already imposed there. Then the spectrum
  $\sigma\left(\Gamma^{*}\right)$ coincides with the set
  \begin{equation}
    \label{eq:setJ}
    \Delta = \left\{ k^2 \in \Reals\,\left|\,\det\left(\id-\tU\right)=0\right.\right\}.
  \end{equation}
\end{lem}

\begin{proof}
  We mention that imposition of the additional vertex conditions
  makes the problem overdetermined.  In most circumstances the set $\Delta$
  will be empty.  The operator $H$ is still symmetric but no longer
  self-adjoint, because its domain is too narrow.

  Denote by $\tilde{\Gamma}$ the graph with the leads attached.  Let
  $k^2\in\sigma\left(\Gamma^{*}\right)$ and let $f$ be the
  corresponding eigenfunction on $\Gamma^{*}$.  Then $f$ can be
  extended to the leads by zero.  It will still satisfy the vertex
  conditions of $\tilde{\Gamma}$ and will therefore satisfy
  \eqref{eq:full_graph_scatter} with
  $\vec{c}^{\,in}=\vec{c}^{\,out}=0$ and $\vec{a}\neq0$. The last
  $2|\Edges|$ equations of \eqref{eq:full_graph_scatter} imply
  $\det(\id-\tU)=0$.

  In the other direction, let $k^2\in \Delta$.  Choose
  $\vec{a}\in\Ker(\id-\tU)$.  We see that equation
  \eqref{eq:full_graph_scatter} is satisfied with the chosen $\vec{a}$
  and with $\vec{c}^{\,in}=\vec{c}^{\,out}=0$.  These coefficients
  describe a function on $\tilde{\Gamma}$ which vanishes completely on
  the leads and therefore its restriction to $\Gamma$ satisfies
  Dirichlet boundary conditions on the marked vertices by continuity.
  This implies $k^2\in\sigma\left(\Gamma^{*}\right)$.
\end{proof}

\begin{cor}
  \label{cor:discrete_special_k_values}
  The set $\Delta$ is discrete.
\end{cor}
\begin{proof}
  The corollary is immediate since
  $\Delta\subset\sigma\left(\Gamma\right)$,
  which is discrete.
\end{proof}

\begin{lem}
  $S\left(k\right)$ is a meromorphic function which is analytic on the
  real line.
\end{lem}

\begin{proof}
  The blocks of the matrix $\Sigma(k)$ in equation
  \eqref{eq:full_graph_scatter0} are meromorphic (see
  \cite{KosSch_jpa99}, Theorem 2.1 and the discussion following it),
  therefore all the blocks of the matrix $Q$ are meromorphic on the
  entire complex plane.  Since the set $\Delta$ on which the matrix
  $\id-\tU$ is singular is discrete, equation \eqref{eq:S_explicit}
  defines a meromorphic function.  To show that $S(k)$ in fact does
  not have singularities on the real line, we observe, that, for $k
  \in\Reals \setminus \Delta$ we have shown that $S(k)$ defined by
  \eqref{eq:S_explicit} is unitary.  Therefore $S(k)$ remains bounded
  as we approach the ``bad'' set $\Delta$ and the singularities are
  removable.  Theorem~\ref{thm:scat_mat} gives a prescription for
  computing the correct value of $S(k)$ for $k \in \Delta$.
\end{proof}

We now examine the $k$-dependence of the eigenvalues of $S(k)$.
To avoid technical difficulties we restrict our attention to the case
when only $\alpha=0$ (Neumann) or $\alpha=\infty$ (Dirichlet) are
allowed as coefficients of the $\delta$-type vertex conditions,
equation~\eqref{eq:delta_deriv}.  In this case the matrix $\Sigma(k)$
described in section~\ref{sec:scat_mat} is independent of $k$ making
calculations easier.  The general case can be treated using methods of
\cite{BolEnd_incol08}, however we will not need it for applications.

\begin{lem}
  \label{lem:counterclockwise}
  Let every vertex of the graph $\Gamma$ have either Neumann or
  Dirichlet condition imposed on it.  Then the eigenvalues of $S(k)$
  move counterclockwise on the unit circle, as $k$ increases.
\end{lem}

\begin{proof}
  Let $e^{i\varphi}$ be an eigenvalue of $S$ with the normalized
  eigenvector $v$.  Differentiating the normalization condition
  $v^*v=1$ with respect to $k$ we get
  \begin{equation}
    \label{eq:v_vdot_orth}
    v^* \dot{v} = 0.
  \end{equation}
  Now we take the derivative of $Sv=e^{i\varphi}v$ with respect to $k$
  to get
  \begin{equation*}
    \left(\frac{{\rm d}}{{\rm dk}}S\right)v + S\dot{v}
    = i\dot{\varphi}e^{i\varphi}v + e^{i\varphi}\dot{v}.
  \end{equation*}
  We multiply the above equation on the left with
  $v^{*}S^{*}=e^{-i\varphi}v^{*}$ and use $v^{*}v=1$ and
  equation~\eqref{eq:v_vdot_orth} to obtain:
  \begin{equation*}
    v^{*}S^{*}\left(\frac{{\rm d}}{{\rm dk}}S\right)v = i\dot{\varphi}.
  \end{equation*}
  Thus we need to show that $\frac{1}{i}S^{*}\frac{{\rm d}}{{\rm
      dk}}S$ is positive definite.  Comparing equations
  \eqref{eq:full_graph_scatter0} and \eqref{eq:full_graph_scatter} and
  using that $\Sigma$ is $k$-independent, we obtain that
  \begin{equation*}
    R(k)=R(0), \qquad T_o(k)=T_o(0), \qquad
    T_i(k) = \rme^{\rmi kL} T_i(0), \qquad \tU(k) = \rme^{\rmi kL} \tU(0).
  \end{equation*}
  Differentiating the latter two matrices with respect to $k$ produces
  \begin{equation*}
    \dot{\tU}=iL\tU, \qquad \dot{T_{i}}=iLT_{i}.
  \end{equation*}
  We can now differentiate equation \eqref{eq:C-definition} to obtain
  \begin{equation}
    \label{eq:Cdot}
    (1-\tU)\dot{C} = \dot{T_i} + \dot{\tU}C = iL(T_{i} + \tU C) = iLC,
  \end{equation}
  where we used \eqref{eq:C-definition} again in the final step.

  For the matrix in question we now obtain
  \begin{equation*}
    S^{*}\frac{{\rm d}}{{\rm dk}}S
    = \left(R^{*}+C^{*}T_{o}^{*}\right) T_{o}\dot{C}
    =-T_{i}^{*}\tU\dot{C}+C^{*}\left(\id-\tU^{*}\tU\right)\dot{C},
  \end{equation*}
  where equations \eqref{eq:unitarity_of_Q} have been used in the
  second step.  Using $-T_i^* = C^*(\tU^*-\id)$ which is a conjugate of
  \eqref{eq:C-definition}, we obtain
  \begin{equation*}
    S^{*}\frac{{\rm d}}{{\rm dk}}S
    = C^* \left(\tU^*\tU - \tU + \id - \tU^{*}\tU\right)\dot{C}.
  \end{equation*}
  Using \eqref{eq:Cdot} this simplifies to
  \begin{equation*}
    S^{*}\frac{{\rm d}}{{\rm dk}}S  = iC^{*}LC.
  \end{equation*}
  Since $L$ is diagonal with positive entries we conclude that
  $\frac{1}{i}S^{*}\frac{{\rm d}}{{\rm dk}}S$ is positive definite.
\end{proof}

We end this section by stating a result known as the inside-outside
duality, which relates the spectrum of the compact graph, $\Gamma$, to
the eigenvalues of its scattering matrix, $S(k)$.  This is a well
known result, mentioned already in \cite{KotSmi_ap99}. We bring it here with a small modification, related to
the already mentioned set, $\Delta$.

\begin{prop}
  \label{pro:inside-outside}
  The spectrum of
  $\Gamma$ is $\Delta\cup\left\{ k\,|\,\det\left(\id-S\right)=0\right\} $.
\end{prop}

\begin{proof}
  We remind the reader that when a lead is attached to a (marked)
  vertex, the new vertex conditions are also of $\delta$-type with the
  same value of the constant $\alpha_v$.  The conditions at the
  vertices that are not marked remain unchanged.

  Let $k$ be such that $\det\left(\id-S\right)=0$.  Let $\vec{c}$ be
  the corresponding eigenvector, $S\vec{c} = \vec{c}$.  Letting
  $\vec{c}^{\,in} = \vec{c}$ we find $\vec{c}^{\,out} = \vec{c}$ and
  $\vec{a}$ according to theorem~\ref{thm:scat_mat}.  The
  corresponding generalized eigenfunction $\tf$ satisfies correct
  vertex conditions at all the non-marked vertices.  It is also
  continuous at the marked vertices and satisfies
  \begin{equation}
    \label{eq:cond_deriv_leads}
    \sum_{e \in \Edges_v}\frac{d\tf}{dx_e}(v) + \sum_{l\in\Leads_v}
    \frac{d\tf}{dx_l}(v) =\alpha _v \tf(v),
  \end{equation}
  where $\Edges_v$ is the set of the finite edges incident to $v$ and
  $\Leads_v$ is the set of the infinite leads attached to it.
  Referring to \eqref{eq:sol_on_lead} we notice that $c_{v,l}^{\,in} =
  c_{l,v}^{\,out}$ implies that the derivative of $\tf$ on the lead is
  zero.  Therefore equation \eqref{eq:cond_deriv_leads} reduces to the
  corresponding equation on the compact graph.  Thus the restriction
  of $\tf$ to the compact graph satisfies vertex conditions at all
  vertices and $k^2$ is an eigenvalue of $\Gamma$.  Inclusion $\Delta
  \subseteq \sigma(\Gamma)$ has already been shown in
  lemma~\ref{lem:special_k_values}.

  Conversely, let $k^2$ be an eigenvalue of the compact graph $\Gamma$
  and let $f$ be the corresponding eigenfunction.  Then $f$ can be
  continued onto the leads as $f(v)\cos(k x)$, where $f(v)$ is the
  value of $f$ at the vertex $v$ where the lead is attached to the
  graph.  Comparing to \eqref{eq:sol_on_lead} we see that
  $c_{l,v}^{in} = c_{l,v}^{out} = f(v)/2$.  Therefore the resulting
  extended function is characterized by vectors $\vec{a}$,
  $\vec{c}^{\,in}$ and $\vec{c}^{\,out}$ such that $\vec{c}^{\,in} =
  \vec{c}^{\,out} = S \vec{c}^{\,in}$.  If the function $f$ was
  non-zero on at least one of the marked vertices, $\vec{c}^{\,in}\neq
  0$ is a valid eigenvector of $S(k)$ with eigenvalue 1.  If $f$ is
  zero on all marked vertices, $k\in\Delta$ by
  lemma~\ref{lem:special_k_values}.
\end{proof}

\section{Applications to the nodal domains count}
\label{sec:applications_to_nodal_count}

\subsection{\label{sec:single_lead_dyn}Application for a single lead case.}

We wish to study the nodal count sequence of a certain graph $\Gamma$
by attaching a single lead to one of its vertices. Let $S(k)=e^{i\varphi(k)}$
be the corresponding one dimensional scattering matrix. For each real
$k>0$ there exists a generalized eigenfunction, $f(k;\, x)$, of
the Laplacian with eigenvalue $k^{2}$ on the extended graph, $\tilde{\Gamma}$,
as proved in theorem \ref{thm:scat_mat}. In addition, up to a multiplicative
factor, this function is uniquely determined on the lead, where it
equals \begin{align*}
f_{lead}(k;\, x) & =c^{in}\exp(-ikx)+c^{out}\exp(ikx)=c^{in}\left[\exp(-ikx)+\exp\left(i\left(\varphi(k)+kx\right)\right)\right]\\
 & =c^{in}\exp\left(i\frac{\varphi(k)}{2}\right)\cos\left(kx+\frac{\varphi(k)}{2}\right).\end{align*}
 The positions of the nodal points of this function on the lead are therefore
uniquely defined for every real $k>0$ and given by\begin{equation}
D_{lead}\left(k\right)=\left\{ x\ge0\,\left|\, x\in-\frac{\varphi\left(k\right)}{2k}+\frac{\pi}{2k}+\frac{\pi}{k}\mathbb{Z}\right.\right\} .\label{eq:nodal_points_on_single_lead}\end{equation}
 We exploit this by treating $k$ as a continuous parameter and inspecting
the change in the positions of the nodal points as $k$ increases. Let $x=-\frac{\varphi\left(k\right)}{2k}+\frac{\pi}{2k}+\frac{\pi}{k}n_{x}$
be the position of a certain nodal point on the lead at some value
$k$, i.e., $x\in D_{lead}\left(k\right)$. The direction of movement of
this nodal point is given by

\begin{eqnarray}
\frac{\textrm{d}}{\textrm{dk}}\left(-\frac{\varphi\left(k\right)}{2k}+\frac{\pi}{2k}+\frac{\pi}{k}n_{x}\right) & = & -\frac{k\varphi'(k)-\varphi(k)+\pi+2\pi n_{x}}{2k^{2}}\nonumber \\
 & = & -\frac{1}{k}\left(\frac{\varphi'(k)}{2}+x\right)<0,\label{eq:nodal_dynamics_1_lead}\end{eqnarray}
 where for the last inequality we need to assume that all the vertex
conditions of $\Gamma$ are either of Dirichlet or of Neumann type
in order to use the conclusion of lemma \ref{lem:counterclockwise}
, $\varphi'\left(k\right)>0$. From \eqref{eq:nodal_dynamics_1_lead}
we learn that all nodal points on the lead move towards the graph,
as $k$ increases.

The event of a nodal point arriving to the graph from the lead occurs
at values $k$ for which $0\in D_{lead}\left(k\right)$ and will be
called an \emph{entrance event}. We may use \eqref{eq:nodal_points_on_single_lead}
to characterize these events in terms of the scattering matrix: \begin{equation}
0\in D_{lead}\left(k\right)\Leftrightarrow\mmod_{2\pi}\left(\varphi(k)\right)=\pi.\label{eq:entrance_event_single_lead}\end{equation}
 After such an event occurs, the nodal point from the lead enters
the regime of the graph $\Gamma$ and may change the total number
of the nodal points of $f(k;\, x)$ within $\Gamma$.

Another significant type of events is described by

\begin{equation}
\mmod_{2\pi}\left(\varphi(k)\right)=0.\label{eq:eigenvalue_single_lead}\end{equation}
 Proposition \ref{pro:inside-outside} shows that such an event happens
at a spectral point of $\Gamma$ and at this event, the restriction
of $f(k;\, x)$ on $\Gamma$ equals the corresponding eigenfunction
of $\Gamma$. These events form the whole spectrum of $\Gamma$ if
and only if $\Delta=\emptyset$. This is indeed the case if we choose
to attach the lead to a position where none of the graph's eigenfunctions
vanish (lemma \ref{lem:special_k_values}). In addition, we will assume
in the following discussion that $\Gamma$ has a simple spectrum.
This is needed for the unique definition of the nodal count sequences
and is shown in \cite{Fri_ijm05} to be the generic case for quantum
graphs.

The two types of events described by
\eqref{eq:entrance_event_single_lead} and
\eqref{eq:eigenvalue_single_lead} interlace, as we know from lemma
\ref{lem:counterclockwise} (compare also with theorem
\ref{thm:interlacing_monotone}).  We may investigate the nodal count
of tree graphs by merely considering these two types of events and
their interlacing property. We count the number of nodal points within
$\Gamma$ only at the spectral points to obtain the sequence $\left\{
  \mu_{n}\right\} $. Between each two spectral points we have an
entrance event, during which the number of nodal points within
$\Gamma$ increases by one, as a single nodal point has entered from
the lead into the regime of $\Gamma$. This interlacing between the
increments of the number of nodal points and its sampling gives
$\mu_{n}=n-1$ and $\nu_{n}=\mu_{n}+1=n$.

The above conclusion is indeed true for tree graphs under certain
assumptions (see \cite{AlO_viniti92,PokPryObe_mz96,Sch_wrcm06}).
However, when graphs with cycles are considered, there are other
interesting phenomena to take into account:
\begin{enumerate}
\item In the paragraph above it was taken for granted that at an
  entrance event the nodal points count increases by one. This is
  indeed so if the nodal point which arrives from the lead enters
  exactly one of the edges of $\Gamma$ without interacting with other
  nodal points which already exist on the graph.  However, when the
  lead is attached to a cycle of the graph, the generic behavior is
  either a split or a merge. Assume for simplicity that the attachment
  vertex has degree 3, counting the lead.  A split event happens when
  a nodal point from the lead splits into two nodal points that
  proceed the two internal edges. This will increase the number of
  nodal points on $\Gamma$ correspondingly.  In a merge event the
  entering nodal point merges with another nodal point coming along
  one of the internal edges.  The resulting nodal point proceeds along
  the other internal edge.  The number of nodal points on $\Gamma$
  will not change during such an event.  If a lead is attached to a
  vertex of higher degree the variety of scenarios can be larger.
\item Another type of events that were not considered are ones in
  which a nodal point travels on the graph and reaches a vertex which
  is not connected to the lead. When this vertex belongs to a cycle,
  the generic behavior would be a split or a merge event and would
  correspondingly increase or decrease the number of nodal points
  present inside the graph.
\end{enumerate}
These complications are dealt with in section
\ref{sec:single_cycle_dyn}, where we use the single lead approach to
derive a nodal count formula for graphs which contain a single
cycle. In the following section, \ref{sec:sign_counting}, we consider
a modification of this method --- we attach two leads to a graph and use
the corresponding scattering matrix to express nodal count related
quantities of the graph. Later, in section \ref{sec:dihedral_formula},
we show how the two leads approach yields an exact nodal count formula
for a specific graph.

\subsection{\label{sec:sign_counting}Sign-weighted counting function}

The number of nodal points on a certain edge $e=\left(u,v\right)$ at
an eigenvalue $k_n$ is given by
\begin{equation}
\hspace{-5mm}
  \left\lfloor\frac{k_n L_{ e }}{\pi}\right\rfloor+\frac{1}{2} \left( 1-(-1)^{
      \left\lfloor\frac{k_n L_{e}}{\pi} \right\rfloor}
    \sgn[f_n(u)\,f_n(v)]
  \right ), \
  \label{eq:NNB2}
\end{equation}
where $\lfloor x \rfloor$ stands for the largest integer which is
smaller than $x$, and $f_n$ is the corresponding eigenfunction \cite{GnuSmiWeb_wrm04}.
We infer that the relative sign of the eigenfunction at two chosen points is of particular interest when counting nodal domains.
While the most natural candidates for the two points are end-points of an
edge, the results of this section apply to any two points on a graph.
Denoting these points $x_1$ and $x_2$, we are interested in the sign
of the product $f_n(x_1) f_n(x_2)$, where $f_n$ is the $n$-th
eigenfunction of the graph.  We define the \emph{sign-weighted
  counting function} $N_{x_1,x_2}(k)$ as
\begin{equation}
  \label{eq:swcf_def}
  N_{x_1,x_2}(k) = \# \{k_n \leq k: f_n(x_1) f_n(x_2)>0
  \} - \# \{k_n \leq k: f_n(x_1) f_n(x_2)<0
  \}.
\end{equation}
Using the scattering matrix formalism allows us to obtain the
following elegant formula

\begin{thm}
  \label{thm:swcf}
  Let $\Graph$ be a graph with Neumann or
  Dirichlet vertex conditions and $x_1$ and $x_2$ be points on the
  graph such that no eigenfunction turns to zero at $x_1$ or $x_2$.
  Denote by $S(k)$ the $2\times2$ scattering matrix obtained by attaching
  leads to the points $x_1$ and $x_2$.  Let $\Id_\epsilon$ be the matrix
  \begin{equation}
    \label{eq:Id_eps}
    \Id_\epsilon = \left(
      \begin{matrix}
        1 & -\epsilon\\
        -\epsilon & 1
      \end{matrix}
    \right) = \Id - \epsilon \sigma_1,
  \end{equation}
  where $\sigma_1$ is the first Pauli matrix.  Then
  \begin{equation}
    \label{eq:swcf_formula}
    N_{x_1,x_2}(k) = \frac{1}{\pi} \left[ \lim_{\epsilon\to0} \arg
    \det(\Id_\epsilon - S(k)) - \phi(k)\right],
  \end{equation}
  where $e^{2i\phi(k)} = \det(S(k))$ and a suitable continuous branch
  of the argument is chosen.  The convergence is pointwise everywhere except
  at $k\in\sigma(\Gamma)$ (where $N_{x_1,x_2}$ is discontinuous).
\end{thm}

\begin{proof}
  First of all, we observe that the scattering matrix of a graph with
  only Neumann or  Dirichlet conditions is complex symmetric: $S_{j,k}
  = S_{k,j}$.  This can be verified explicitly by using
  representation \eqref{eq:C-definition}-\eqref{eq:S-definition},
  together with \eqref{eq:full_graph_scatter0} and the fact that the
  matrix $\Sigma(k)$ is real and symmetric under the specified
  conditions.

  For a moment, consider that only one lead is connected, to the point
  $x_1$. Then the events that change the sign of $f(x_1)$ and
  $f(x_2)$, where $f$ is the one-lead scattering solution, are the
  values of $k$ such that (A) a zero comes into the vertex $x_1$ from
  the lead (``Dirichlet events'') or (B) a zero crosses the point
  $x_2$. The former are easy to characterize: they interlace with the
  events when a ``Neumann point'' comes into the vertex $x_1$, which
  happen precisely at $\sigma(\Gamma)$, as discussed in
  section~\ref{sec:single_lead_dyn}.

  Denote by $\kappa$ the value of $k$ when a zero crosses the point
  $x_2$ where the lead is \emph{not} attached. Now consider the
  scattering system when both leads are attached, at points $x_1$ and
  $x_2$.  At $k=\kappa$ the one-lead scattering solution $f(x)$ can be
  continued to the second lead by setting it to vanish on the entire
  lead.  This would create a valid two-lead solution with $c_2^{in} =
  c_2^{out} = 0$.  By inspecting \eqref{eq:c_out} we conclude that the
  vector $(1,0)^T$ is therefore an eigenvector of the two-lead
  scattering matrix $S(\kappa)$.  This happens whenever the matrix
  $S(k)$ is diagonal.  We conclude that the events of type (B) happen
  in a one-lead scattering scenario precisely when the two-lead
  scattering matrix satisfies $S(k)_{1,2} = S(k)_{2,1} = 0$.

  Introducing the notation,
  \begin{align}
    \zeta(k) & =\det(I-S(k))
    \label{eq:zeta_def} \\
    \tau(k) & =S(k)_{1,2}=S(k)_{2,1},
    \label{eq:tau_def}
  \end{align}
  we can summarize the earlier discussion as follows.  The eigenvalues
  of the graph are given by the zeros of $\zeta(k)$ (see
  proposition~\ref{pro:inside-outside}), and the relative sign of the
  $n$-th eigenfunction, $\sgn[f_n(x_1)\,f_n(x_2)]$, is equal to the parity of the total number of
  zeros of $\zeta$ and $\tau$ that are strictly less than $k_n$.  Note
  that the condition that no eigenfunction is zero on $x_1$ or $x_2$
  implies that the set $\Delta$ in
  proposition~\ref{pro:inside-outside} is empty and that the zeros of
  the functions $\zeta$ and $\tau$ are distinct.

  Applying complex conjugation to $\zeta(k)$ we obtain
  \begin{equation}
    \zeta(k)^{*}=\det(S^{*})\zeta(k)=\zeta(k)/\det(S).
    \label{eq:zeta_conj}
  \end{equation}
  Similarly, using the explicit formula for the inverse of a
  $2\times2$ matrix together with the unitarity of $S$, we obtain for $\tau(k)$
  \begin{equation}
    \tau(k)^{*}=-\det(S^{*})\tau(k)=-\tau(k)/\det(S).
    \label{eq:tau_conj}
  \end{equation}
  These relations allow us to represent $\zeta(k)=r(k)e^{i\phi(k)}$
  and $\tau(k)=\frac{i}{2}t(k)e^{i\phi(k)}$, when recalling that $e^{2i\phi(k)} = \det(S(k))$.

  We now evaluate
  \begin{align*}
    \det(\Id_\epsilon - S) &= \det(\Id-S) - \epsilon(S_{1,2}+S_{2,1}) -
    \epsilon^2\\
    &= \zeta(k) - \epsilon \tau(k) - \epsilon^2 = (r(k) -
    \rmi \epsilon t(k)) e^{\rmi\phi(k)} - \epsilon^2,
  \end{align*}
  and, therefore,
  \begin{equation}
    \label{eq:rat_dets}
    \det(\Id_\epsilon - S)e^{-i\phi(k)} = (r(k) -
    \rmi \epsilon t(k)) + o(\epsilon).
  \end{equation}
  It is now clear that, when $r(k) \neq 0$ (i.e. when
  $k\not\in\sigma(\Gamma)$), the $\epsilon\to0$ limit of the above
  ratio is a non-zero real number and therefore its argument is an integer
  multiple of $\pi$.

  To evaluate this integer we focus on the values of $k$ when $z(k) =
  r(k) - \rmi \epsilon t(k)$ crosses the line $\Re(z) = 0$.  When the crossing is in the counter-clockwise direction, the integer above increases, and otherwise it decreases. The counter-clockwise versus clockwise direction of the crossing is
  decided exclusively by the sign of the ratio $r(k_n-0) / t(k_n-0)$,
  which coincides with the parity of the total number of zeros of the
  two functions.  This, in turn, has been shown to coincide with the
  relative sign of the eigenfunction.
\end{proof}

\begin{remark}
  It is interesting to compare the above formula for the sign-weighted
  counting function with the corresponding formula for the more commonly used spectral counting function,
  \begin{equation}
    \label{eq:cf_def}
    N(k) = \# \{k_n \leq k\}.
  \end{equation}
  Under conditions of theorem~\ref{thm:swcf} the counting function
  $N(k)$ can be represented as
  \begin{equation}
    \label{eq:cf_formula}
    N(k) = \frac{1}{\pi} \left[\lim_{\epsilon\to0} \arg
    \det(\Id-\epsilon\Id - S)-\phi(k)\right],
  \end{equation}
  Combining the two we can obtain a counting function that counts only
  the eigenvalues whose eigenfunctions have differing signs at $x_1$
  and $x_2$,
  \begin{equation}
    \label{eq:negcf_formula}
    N_{x_1,x_2}^{-}(k) := \# \{k_n \leq k: f_n(x_1) f_n(x_2)<0 \}
    = \frac{1}{2\pi} \lim_{\epsilon\to0} \arg
    \frac{\det(\Id-\epsilon\Id - S)}{\det(\Id-\epsilon\sigma_1 - S)}.
  \end{equation}
\end{remark}

\subsection{\label{sec:dihedral_formula}Using two leads to derive an exact nodal
count formula}

In the current section we derive a nodal count formula for the graph,
$\Gamma$, given in figure \ref{fig:original_dihedral}.
\begin{figure}[ht]
 \includegraphics[scale=0.6]{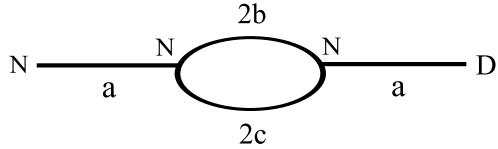}

\caption{\label{fig:original_dihedral}The quantum graph whose nodal count
we compute. The lengths of the edges and the vertex conditions are
indicated. \emph{D} stands for Dirichlet vertex conditions and \emph{N}
for Neumann ones.}

\end{figure}
This graph is a member of an isospectral pair, as described in \cite{BPB09}.
The isospectral twin of this graph is the graph shown in figure \ref{fig:tree_dihedral}.
\begin{figure}[ht]
 \includegraphics[scale=0.6]{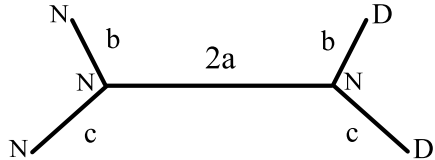}

\caption{\label{fig:tree_dihedral}The tree graph which is isospectral to the
graph in figure \ref{fig:original_dihedral}. The lengths of the edges
and the vertex conditions are indicated.}

\end{figure}

Examining the topology of each of the graphs according to \eqref{eq:nodal_domains_bound}
tells us that the tree graph has the nodal count $\nu_{n}=n$ and
the nodal count of the graph $\Gamma$, which has a cycle, obey the
bounds $n-1\leq\nu_{n}\leq n$. It was claimed in \cite{BOS08} that
the nodal count of $\Gamma$ is \begin{equation}
\nu_{n}=n-\frac{1}{2}-\frac{1}{2}\left(-1\right)^{\left\lfloor \frac{b+c}{a+b+c}n\right\rfloor }.\label{eq:dihedral_old_formula}\end{equation}
 This formula was not proved there, but rather a numerical justification
was given. We present here a proof for the following theorem

\begin{thm}\label{thm:dihedral_formula}

Let $a,b,c$ be positive real numbers such that $\frac{b}{c}\notin\mathbb{Q}$
and $\frac{a}{b+c}\notin\mathbb{Q}$. Let $\Gamma$ be the graph described
in figure \ref{fig:original_dihedral}. Then the nodal points count
sequence of $\Gamma$ is

\begin{equation}
\mu_{n}=n+\mmod_{2}\left(\left\lfloor \frac{b+c}{a+b+c}n\right\rfloor \right),\label{eq:dihedral_nodal_points_formula}\end{equation}

and the nodal domains count sequence of $\Gamma$ is:

\begin{equation}
\nu_{n}=\begin{cases}
n & n\leq\left\lfloor \frac{a}{b+c}\right\rfloor +1\\
n-1+\mmod_{2}\left(\left\lfloor \frac{b+c}{a+b+c}n\right\rfloor \right) & n>\left\lfloor \frac{a}{b+c}\right\rfloor +1\end{cases},\label{eq:dihedral_nodal_domains_formula}\end{equation}

\end{thm}

\begin{remark}Note that for eigenvalues high enough in the spectrum,
the exact nodal count \eqref{eq:dihedral_nodal_domains_formula} coincides
with the previously known numeric result, \eqref{eq:dihedral_old_formula}.
\end{remark}

The method of proof of the formulas \eqref{eq:dihedral_nodal_points_formula}, \eqref{eq:dihedral_nodal_domains_formula} involves attaching leads to the graph and keeping track of the nodal points dynamics with respect to an increment of the spectral parameter. This specific example both presents the ability to derive an exact formula and also demonstrates the technical complications that may arise while using this method.

\subsubsection{\label{sec:dihedral_proof_outline}A brief outline of the proof}
Define a graph with two vertices connected by two edges of lengths
$2b$ and $2c$. This graph is actually a cycle of length $2b+2c$.
Connect two leads to the vertices of this graph to obtain the graph
in figure \ref{fig:basic_loop_with_leads}. Denote this graph by $\tilde{\Gamma}$
and notice that $\Gamma$ as a metric graph is a subgraph of $\tilde{\Gamma}$.
The graph $\tilde{\Gamma}$ has a symmetry of reflection along an
axis which passes through the middle of the graph. We will exploit
this symmetry in the next section.

\begin{figure}[h]
 \includegraphics[scale=0.6]{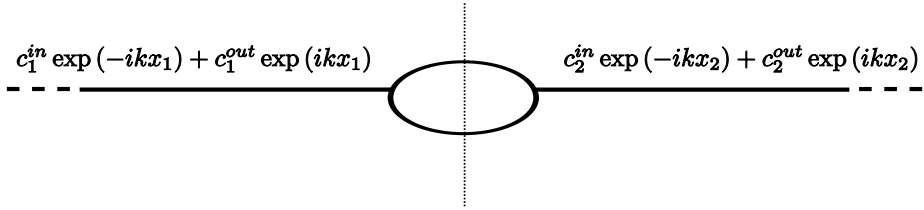}

\caption{\label{fig:basic_loop_with_leads}The graph $\tilde{\Gamma}$ and
its symmetry axis (dotted). The values of its generalized eigenfunction
are specified on the leads.}

\end{figure}

The Laplacian on $\tilde{\Gamma}$ possesses a continuous spectrum
and each generalized eigenvalue, $k^{2},$ has a two-dimensional generalized
eigenspace, characterized by $\vec{c}^{\,in}\in\mathbb{C}^{2}$ (theorem
\ref{thm:scat_mat}). We will describe a one parameter ($k\in\Reals$)
family of generalized eigenfunctions on $\tilde{\Gamma}$, $f(k;\, x)$.
We thus consider $f(k;\, x)$
as a function on $\tilde{\Gamma}$ which changes continuously with
$k$ - this will be emphasized by the notation $f(k;\,\cdot)$. This
$k$-dependent function would be chosen such that its restriction
on the subgraph $\Gamma$ at $k^{2}\in\sigma\left(\Gamma\right)$
equals the corresponding eigenfunction of $\Gamma$. The strategy
of the proof is to keep track of the number of nodal points of $f(k;\,\cdot)$
as it changes with $k$ and to sample this number at $k^{2}\in\sigma\left(\Gamma\right)$.
We will notice that the nodal points travel continuously from infinity
towards the cycle and we will characterize the dynamics of the nodal
points which enter the cycle. This will allow us to find the change
in the number of nodal points during such entrance events. We will
then calculate the number of eigenvalues which occur between two consequent
entrance events and will combine all those observations to deduce
the nodal count formulas \eqref{eq:dihedral_nodal_points_formula}
and \eqref{eq:dihedral_nodal_domains_formula}.

\subsubsection{Towards a proof of theorem \ref{thm:dihedral_formula}}

Let $\tilde{\Gamma}$ be the graph that is described in section \ref{sec:dihedral_proof_outline}
and appears in figure \ref{fig:basic_loop_with_leads}. A generalized
eigenfunction of $\tilde{\Gamma}$ with eigenvalue $k^{2}$, on the
j$^{th}$ lead, is given by \[
f_{j}(k;\, x_{j})=c_{j}^{in}\exp(-ikx_{j})+c_{j}^{out}\exp(ikx_{j}),\]
 where $j=1,2$ and the coefficients \[
\vec{c}^{\,in}=\left(\begin{array}{c}
c_{1}^{in}\\
c_{2}^{in}\end{array}\right),\,\vec{c}^{\,out}=\left(\begin{array}{c}
c_{1}^{out}\\
c_{2}^{out}\end{array}\right)\]
 are related by\begin{align}
\vec{c}^{\, out} & =S(k)\vec{c}^{\, in}.\label{eq:scatter_basic}\end{align}
 The graph $\tilde{\Gamma}$ obeys a symmetry of reflection along
a vertical axis which passes through the center of the graph. This
reflection symmetry exchanges the two leads of $\tilde{\Gamma}$ and
it implies that its scattering matrix, $S(k)$, commutes with the
matrix \[
\sigma=\left(\begin{array}{cc}
0 & 1\\
1 & 0\end{array}\right).\]
 This, together with the unitarity of $S(k)$ (theorem \ref{thm:scat_mat})
allows us to write it in the form \begin{equation}
S(k)=\left(\begin{array}{cc}
\cos\left(\gamma\left(k\right)\right)e^{i\varphi\left(k\right)} & \sin\left(\gamma\left(k\right)\right)e^{i\left(\varphi\left(k\right)+\nicefrac{\pi}{2}\right)}\\
\sin\left(\gamma\left(k\right)\right)e^{i\left(\varphi\left(k\right)+\nicefrac{\pi}{2}\right)} & \cos\left(\gamma\left(k\right)\right)e^{i\varphi\left(k\right)}\end{array}\right).\label{eq:scattering_matrix}\end{equation}
 \\
 The exact form of $S(k)$ (expressed in terms of the edge lengths
parameters $a,b,c$) can be calculated using \eqref{eq:full_graph_scatter0}
, \eqref{eq:full_graph_scatter} and \eqref{eq:S-definition}.

Following the approach described in section \ref{sec:single_lead_dyn}
we treat $k$ as a continuous parameter and choose $\vec{c}^{\,in}$ to vary
continuously with $k.$ Namely, we choose a certain continuous vector
function $\vec{c}^{\,in}:\left(0,\infty\right)\to\mathcal{\mathrm{\mathbb{C}}}^{2}$.
Relation \eqref{eq:scatter_basic} yields the continuous function
$\vec{c}^{\,out}:\left(0,\infty\right)\to\mathcal{\mathrm{\mathbb{C}}}^{2}$
and both $\vec{c}^{\,in}\left(k\right)$ and $\vec{c}^{\,out}\left(k\right)$
determine $f(k;\,\cdot)$, a function on $\tilde{\Gamma}$ that changes
continuously with $k$. We next describe a specific choice of $\vec{c}^{\,in}\left(k\right)$
that yields a function $f(k;\, x)$ with the following properties
which are convenient for our proof.

\begin{property}\label{property_real}

The values of the function on the leads are real, i.e.,

\[
\,\,\,\, f_{j}(k;\, x_{j})=\overline{f_{j}(k;\, x_{j})}\text{\,\,\, for}\,\, j=1,2.\]

\end{property}

\begin{property}\label{property_anti_phase}

Denote the zeros of the function and of its derivative on the leads
by

\begin{eqnarray*}
D_{j}\left(k\right) & := & \left\{ x_{j}\geq0\,\left|\, f_{j}(k;\, x_{j})=0\right.\right\} \\
N_{j}\left(k\right) & := & \left\{ x_{j}\geq0\,\left|\,\frac{\partial}{\partial x_{j}}f_{j}(k;\, x_{j})=0\right.\right\} .\end{eqnarray*}
 They obey \[
D_{1}\left(k\right)=N_{2}\left(k\right)\,\,\text{and}\,\,\, D_{2}\left(k\right)=N_{1}\left(k\right).\]

\end{property}

The usefulness of these properties is made transparent in the following
proposition.

\begin{prop}\label{prop:anti_phase_sol} Let $a,b,c$ be positive
real numbers, such that $\frac{b}{c}\notin\mathbb{Q}$. Let $\Gamma$
and $\tilde{\Gamma}$ the graphs defined above (with the edge lengths
parameters $a,b,c$).
\begin{enumerate}
\item \label{enu:prop_part1}For each $k\in\Reals$ properties \ref{property_real}
and \ref{property_anti_phase} define a function $f(k;\,\cdot)$ on
$\tilde{\Gamma}$ which is unique up to a multiplication by a scalar
and a reflection along a vertical axis which passes through the middle
of $\tilde{\Gamma}$.
\item \label{enu:prop_part2}The above function, $f(k;\,\cdot)$, can be
chosen to be continuous in $k$.
\item \label{enu:prop_part3}If $k^{2}\in\sigma\left(\Gamma\right)$, the
restriction of the function $f(k;\,\cdot)$ to the graph on $\Gamma$
coincides with the eigenfunction of $\Gamma$ up to reflection.
\end{enumerate}
\end{prop}

The following lemma will aid us in proving the uniqueness of $f(k;\,\cdot)$.

\begin{lem}\label{lem:Delta_is_empty}

Let $b,c$ be positive real numbers, such that $\frac{b}{c}\notin\mathbb{Q}$.
Then the set \[
\Delta=\left\{ k\in\Reals\,\left|\,\det\left(\id-\tU\right)=0\right.\right\} ,\]

that was defined in \eqref{eq:setJ}, lemma \ref{lem:special_k_values},
is an empty set.

\end{lem}

\begin{proof}

Lemma \ref{lem:special_k_values} tells us that $\Delta=\sigma\left(\Gamma^{*}\right)$,
where $\Gamma^{*}$ is a cycle graph with additional Dirichlet conditions
imposed on its two vertices (figure \ref{Fig:only_cycle}).

\begin{figure}[h]
 \includegraphics[scale=0.6]{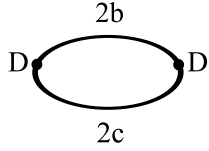}

\caption{The graph $\Gamma^{*}$, which obeys $\Delta=\sigma\left(\Gamma^{*}\right)$.}

\label{Fig:only_cycle}
\end{figure}

Assume that $k$ is in the spectrum of $\Gamma^{*}$. The corresponding
eigenfunction should then be of the form $\sin\left(kx\right)$ on
each of the edges (up to a multiplicative factor). The Dirichlet boundary
conditions imply that $\sin\left(2bk\right)=\sin\left(2ck\right)=0$
and therefore $b$ and $c$ both belong to the set $\frac{\pi}{2k}\mathbb{N}$.
This means that $\frac{b}{c}\in\mathbb{Q}$ and contradicts the assumption.

\end{proof}

\begin{proof} of proposition \ref{prop:anti_phase_sol}

Let $k\in\Reals$. Let $f(k;\,\cdot)$ be a generalized eigenfunction
of $\tilde{\Gamma}$ which obeys the properties \ref{property_real}
and \ref{property_anti_phase}. From property \ref{property_real}
we conclude that \[
c_{j}^{in}=\overline{c_{j}^{out}}\,\,\text{for }j=1,2.\]
 Thus, for a suitable $c_{j}$ and $\alpha_{j}$,\begin{align}
c_{j}^{in} & =c_{j}\exp\left(i\alpha_{j}\right)\nonumber \\
c_{j}^{out} & =c_{j}\exp\left(-i\alpha_{j}\right).\label{eq:c_j_coeffs}\end{align}
 We plug this in the expression for the values of $f$ on the leads
\begin{align*}
f_{j}(k;\, x_{j}) & =c_{j}^{in}\exp(-ikx_{j})+c_{j}^{out}\exp(ikx_{j})=2c_{j}\cos(\alpha_{j}-kx_{j}),\end{align*}
 and obtain \begin{align*}
D_{j}\left(k\right) & =\left\{ x_{j}\geq0\,\left|\, f_{j}(k;\, x_{j})=0\right.\right\} \\
 & =\left\{ x_{j}\geq0\,\left|\, x_{j}\in\frac{\alpha_{j}}{k}+\frac{\pi}{2k}+\frac{\pi}{k}\mathbb{Z}\right.\right\} \\
\\N_{j}\left(k\right) & =\left\{ x_{j}\geq0\,\left|\,\frac{\partial}{\partial x_{j}}f_{j}(k;\, x_{j})=0\right.\right\} \\
 & =\left\{ x_{j}\geq0\,\left|\, x_{j}\in\frac{\alpha_{j}}{k}+\frac{\pi}{k}\mathbb{Z}\right.\right\} .\end{align*}
 Property \ref{property_anti_phase} now translates to \begin{equation}
\alpha_{2}=\alpha_{1}+\frac{\pi}{2}.\label{eq:alpha_coeffs}\end{equation}

We use \eqref{eq:scattering_matrix},\eqref{eq:c_j_coeffs} and \eqref{eq:alpha_coeffs}
and plug them in \eqref{eq:scatter_basic} to get equations on $c_{j},\alpha_{j}$.
There are two possible solutions, which describe two functions that
are the same up to a reflection along a vertical axis which passes
through the middle of $\tilde{\Gamma}$. One of the solutions reads

\begin{eqnarray}
c_{1}^{in}(k) & = & \frac{1}{2}\exp\left(-i\frac{\varphi\left(k\right)}{2}\right)\label{eq:sol_relation_1}\\
c_{2}^{in}(k) & = & \begin{cases}
\frac{\cos\left(\gamma\left(k\right)\right)-1}{2\sin\left(\gamma\left(k\right)\right)}\exp\left(-i\frac{\varphi\left(k\right)-\pi}{2}\right) & \gamma\left(k\right)\notin2\pi\mathbb{Z}\\
0 & \gamma\left(k\right)\in2\pi\mathbb{Z}\end{cases},\label{eq:sol_relation_2}\end{eqnarray}

and the corresponding function is given on the leads by

\begin{eqnarray*}
f_{1}(k;\, x_{1}) & = & \cos\left(\frac{\varphi\left(k\right)}{2}+kx_{1}\right)\\
f_{2}(k;\, x_{2}) & = & \begin{cases}
\frac{\cos\left(\gamma\left(k\right)\right)-1}{\sin\left(\gamma\left(k\right)\right)}\cos\left(\frac{\varphi\left(k\right)-\pi}{2}+kx_{2}\right) & \gamma\left(k\right)\notin2\pi\mathbb{Z}\\
0 & \gamma\left(k\right)\in2\pi\mathbb{Z}\end{cases}.\end{eqnarray*}
 Note that $f_{2}(k;\,\cdot)$ is continuous in $k$. In addition,
$f_{1}(k;\,\cdot)$ and $f_{2}(k;\,\cdot)$ that are given above can
be multiplied by any $k$-continuous scalar function to yield an appropriate
solution which is also continuous in $k$. This proves that $f(k;\,\cdot)$
is uniquely defined on the leads and also $k$-continuous there. It
is left to show the same for the values of $f(k;\,\cdot)$ on the
cycle. Theorem \ref{thm:scat_mat} implies that $f(k;\,\cdot)$ may
have multiple values on the cycle only for $k^{2}\in\Delta$. However,
since $\Delta=\emptyset$ (lemma \ref{lem:Delta_is_empty}), this
cannot happen and $f(k;\,\cdot)$ is uniquely defined on the cycle.
In addition, the values of $f(k;\,\cdot)$ on the cycle are determined
by equation \eqref{eq:avec_in_terms_of_cinvec}, which shows that
these values are continuous in $k$, due to the reversibility of $\id-\tilde{U}$
and the $k$-continuity of $\vec{c}^{\,in}$.

We start proving part \ref{enu:prop_part3} of the proposition by
assuming that $k^{2}\in\sigma\left(\Gamma\right)$. We have that there
exists a real eigenfunction with eigenvalue $k^{2}$ on $\Gamma$.
We fix a function $g(k;\,\cdot)$ on $\tilde{\Gamma}$ to equal this
eigenfunction when restricted on $\Gamma.$ Then the values of this
function, $g(k;\,\cdot)$, can be uniquely continued so that it is
defined on the whole of $\tilde{\Gamma}$. It is easy to verify that
the obtained function obeys properties \ref{property_real} and \ref{property_anti_phase}
and we conclude from the proof of part \ref{enu:prop_part1} of the
proposition that it is equal to $f(k;\,\cdot)$ up to a multiplication
by a scalar or a reflection.

\end{proof}

Proposition \ref{prop:anti_phase_sol} shows that there are only two
$k$-continuous functions, $f(k;\, x)$, which obey the properties
\ref{property_real} and \ref{property_anti_phase}. We call such
a function a \emph{real contra-phasal solution}, due to the properties
that it has. These functions will be used to prove theorem \ref{thm:dihedral_formula}.
We carry on by stating a few lemmas which describe the dynamical properties
of the nodal points of such a real contra-phasal solution.

\begin{lem}\label{lem:np_move_towards_cycle}

The nodal points of a real contra-phasal solution move on the leads
towards the cycle as $k$ increases.

\end{lem}

\begin{proof}

While proving proposition \ref{prop:anti_phase_sol} we have showed
that one of the real contra-phasal solutions has the following values
on the leads

\begin{eqnarray}
f_{1}(k;\, x_{1}) & = & \cos\left(\frac{\varphi\left(k\right)}{2}+kx_{1}\right)\nonumber \\
f_{2}(k;\, x_{2}) & = & \begin{cases}
\frac{\cos\left(\gamma\left(k\right)\right)-1}{\sin\left(\gamma\left(k\right)\right)}\cos\left(\frac{\varphi\left(k\right)-\pi}{2}+kx_{2}\right) & \gamma\left(k\right)\notin2\pi\mathbb{Z}\\
0 & \gamma\left(k\right)\in2\pi\mathbb{Z}\end{cases}.\label{eq:real_contra_phasal_on_leads}\end{eqnarray}

The positions of its nodal points on the leads are therefore given
by \begin{eqnarray}
D_{1}\left(k\right) & = & \left\{ x_{1}\ge0\,\left|\, x_{1}\in-\frac{\varphi\left(k\right)}{2k}-\frac{\pi}{2k}+\frac{\pi}{k}\mathbb{Z}\right.\right\} \nonumber \\
D_{2}\left(k\right) & = & \left\{ x_{2}\ge0\,\left|\, x_{2}\in-\frac{\varphi\left(k\right)}{2k}+\frac{\pi}{k}\mathbb{Z}\right.\right\} .\label{eq:positions_of_nodal_points}\end{eqnarray}

Let $x\left(k\right)=-\frac{\varphi\left(k\right)}{2k}+\frac{\pi}{2k}+\frac{\pi}{k}n_{x}$
be the position of a certain nodal point on the first lead at the value
$k$, i.e., $x\left(k\right)\in D_{1}\left(k\right)$. The direction
in which this nodal point travel on the first lead is given by\begin{eqnarray}
x'\left(k\right)=\frac{\textrm{d}}{\textrm{dk}}\left(-\frac{\varphi\left(k\right)}{2k}+\frac{\pi}{2k}+\frac{\pi}{k}n_{x}\right) & = & -\frac{k\varphi'(k)-\varphi(k)+\pi+2\pi n_{x}}{2k^{2}}\nonumber \\
 & = & -\frac{1}{k}\left(\frac{\varphi'(k)}{2}+x\right).\label{eq:nodal_dyn_two_leads}\end{eqnarray}
 A simple calculation based on \eqref{eq:scattering_matrix} gives
\[
\det S\left(k\right)=\exp\left(i2\varphi\left(k\right)\right).\]
 Denoting the eigenvalues of $S\left(k\right)$ by $\exp\left(i\varphi_{1}\left(k\right)\right),\exp\left(i\varphi_{2}\left(k\right)\right)$,
we have that $\varphi\left(k\right)=\varphi_{1}\left(k\right)+\varphi_{2}\left(k\right)$
and can therefore conclude from lemma \ref{lem:counterclockwise}
that $\varphi'\left(k\right)>0$. Plugging this in \eqref{eq:nodal_dyn_two_leads}
together with $x\ge0$ shows that $x'\left(k\right)>0$. We thus get
that all nodal points on the first lead move towards the cycle, as
$k$ increases. A similar derivation leads to the same conclusion
for the nodal points on the second lead. The second real contra-phasal
solution is a reflection of the one mentioned above and therefore
its nodal points obviously also move towards the cycle.

\end{proof}

\begin{lem}\label{lem:np_at_vertex}

Let $k$ be a value at which a nodal point is positioned on a vertex
of $\tilde{\Gamma}$. The following scenarios exist for the dynamics
of the mentioned nodal point.
\begin{enumerate}
\item The nodal point had arrived to the vertex from a lead. Then, upon
entering the cycle the nodal point will either split into two nodal
points or merge with another nodal point arriving from the cycle.
The set of $k$ values at which these events happen is $\left\{ k_{p}^{*}\right\} _{p=1}^{\infty}:=\left\{ \frac{\pi}{2b+2c}p\right\} $.
The split events happen at $k\in\left\{ k_{1}^{*},k_{3}^{*},k_{5}^{*},...\right\} $
and the merge events at $k\in\left\{ k_{2}^{*},k_{4}^{*},k_{6}^{*},...\right\} $.
\item No nodal point arrives to the vertex from the lead during this event.
The nodal point had therefore arrived to the vertex from the cycle.
It will just flow to the other edge of the cycle. These events happen
at $k$ values for which $\gamma\left(k\right)\in\pi\mathbb{Z}$.
\end{enumerate}
\end{lem}

\begin{proof}

When a nodal point enters the cycle from one of the leads, say the
first one, $f_{1}(k_{p}^{*};\,0)=0$, and we have from property \ref{prop:anti_phase_sol}
that on the second lead $\frac{\partial}{\partial x_{2}}f_{2}(k_{p}^{*};\,0)=0$.
We therefore have that the restriction of $f$ to the cycle during
such an event is equal to an eigenfunction of a single edge of length
$2b+2c$ with Dirichlet vertex conditions at its endpoints.
This implies that the entrance events occur at $k_{p}^{\ast}=\frac{\pi}{2b+2c}p$.
These events are of two types (explanation follows):
\begin{enumerate}
\item At $k_{2m-1}^{\ast}=\frac{\pi}{2b+2c}\left(2m-1\right)$ the entering
nodal point splits into two new nodal points which continue to move
in the cycle. Hence the total number of nodal points increases by
one.
\item At $k_{2m}^{\ast}=\frac{\pi}{2b+2c}2m$ the entering nodal point merges
with another nodal point coming towards it from the cycle. Hence the
total number of nodal points decreases by one (see figure \ref{fig:merge_event}).
\end{enumerate}
During an entrance event, $k=k_{p}^{*}$, the nodal point is positioned
on a vertex of $\tilde{\Gamma}$ and about to enter the cycle. We
observe that the number of nodal points on the cycle must be even.
This implies that at the entrance event the mentioned nodal point
either merges with another nodal point form the cycle (so that the
number of nodal points on the cycle remains unchanged), or splits
into two nodal points (which increases this number by two). The occurrence
of a split or a merge event is determined by the values of $f(k_{p}^{*};\,\cdot)$
restricted on the cycle. As mentioned before, this restriction is
an eigenfunction on the edge of length $2b+2c$ and it is therefore
equals $\sin\left(\frac{p\pi}{2b+2c}x\right)$ up to a multiplicative
scalar. For an even value of $p$, this function has opposite signs
in the vicinity of the endpoints of the edge. This means that when
the nodal point is located exactly on the vertex of $\tilde{\Gamma}$,
the two nodal domains of $f(k_{p}^{*};\,\cdot)$ on the cycle which
are bounded by this nodal point have opposite signs. (see figure \ref{fig:merge_event}-during).

\begin{figure}[h]
 before \includegraphics[scale=0.7]{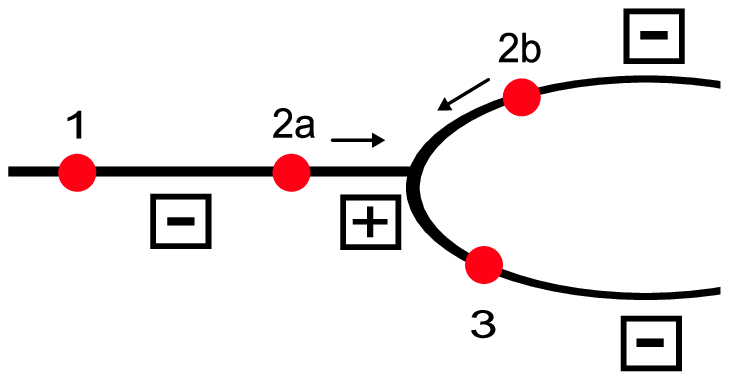}

\vspace{0.4in}

during \includegraphics[scale=0.7]{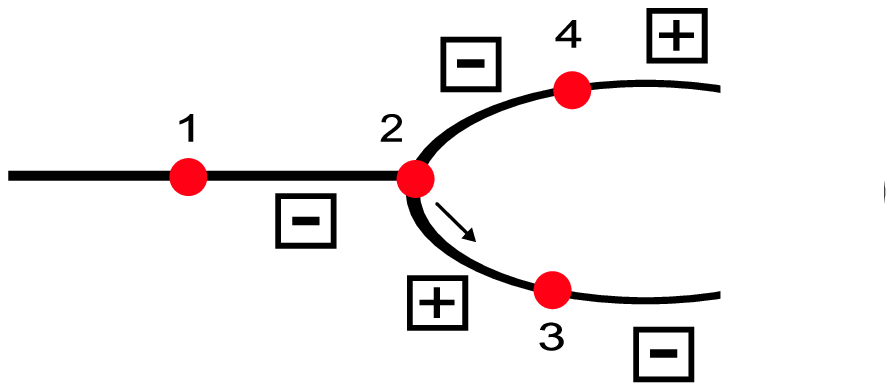}

\vspace{0.4in}

after \includegraphics[scale=0.7]{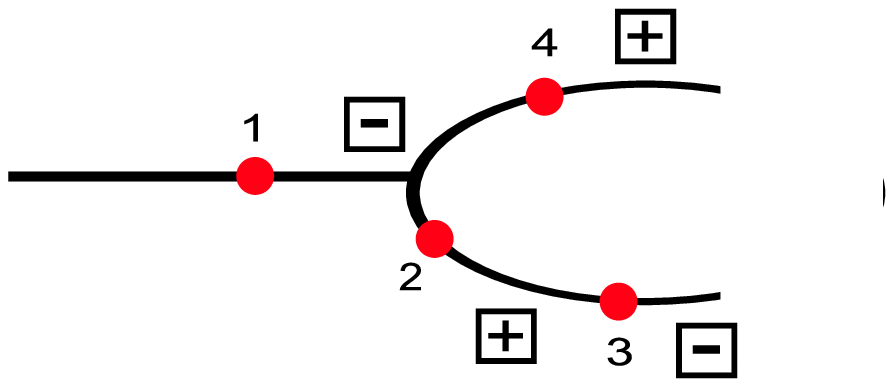}

\caption{A description of a merge event. The nodal points (2a) and (2b) merge
and become the nodal point (2). The signs of the nodal domains of
$f^{(k)}$ are marked with squares.\label{fig:merge_event}}

\end{figure}

However, a short while before this event, the neighborhood of this
vertex was contained in a single nodal domain with a definite sign.
The $k$-continuity of the solution implies that this is possible
only if a short while before the event there was another nodal point
in the vicinity of the vertex that has disappeared while merging with
the nodal point at the vertex (see figure \ref{fig:merge_event}-before)
. A similar reasoning shows that split events occur for odd $p$ values.

We have treated by now the possibility that the nodal point at the
vertex had arrived from the lead. It might also happen that $f(k;\,\cdot)$
equals zero at a vertex of $\tilde{\Gamma}$ when $f(k;\,\cdot)$
vanishes on the lead which is connected to that vertex. For the real
contra-phasal solution given in \eqref{eq:real_contra_phasal_on_leads}
this happens exactly at $\gamma\left(k\right)=\pi\mathbb{Z}$. This
event would happen only on vertex number two for that solution (and
on vertex number one for the reflected solution). These events do
not change the number of nodal points on the graph, and therefore
we do not need to keep track of them.

\end{proof}

\begin{lem}

Let $a,b,c$ be positive real numbers such that $\frac{b}{c}\notin\mathbb{Q}$
and $\frac{a}{b+c}\notin\mathbb{Q}$ and $\Gamma$, $\tilde{\Gamma}$
be the graphs described above. The number of nodal points on $\Gamma$
of a real contra-phasal solution on $\tilde{\Gamma}$ is increased
by one at $k$ such that $k^{2}\in\sigma\left(\Gamma\right)$.

\end{lem}

\begin{proof}

When $k$ equals an eigenvalue of $\Gamma$, the solution restricted
on $\Gamma$ equals an eigenfunction of $\Gamma$, i.e., either $f_{1}(k;\, a)=0$
or $f_{2}(k;\, a)=0$. A nodal point is therefore positioned on the
boundary of $\Gamma$, and from lemma \ref{lem:np_move_towards_cycle}
we deduce that this nodal point moves towards the cycle, increasing
by one the number of nodal points on $\Gamma$. It is only left to
verify that there is no simultaneous split or merge events which further
change the total number of nodal points. Namely, we show that $\left\{ k_{p}^{*}\right\} $
and $\sigma\left(\Gamma\right)$ are disjoint sets. Assume the contrary:
$k_{p}^{*}\in\sigma\left(\Gamma\right)$ for some $p$. By definition,
$f_{j}(k_{p}^{*};\,0)=0$ for $j\in\left\{ 1,2\right\} $. Assume
without loss of generality that $f_{1}(k_{p}^{*};\,0)=0$. Then, since
$k_{p}^{*}\in\sigma\left(\Gamma\right)$ we also have that either
$f_{1}(k_{p}^{*};\, a)=0$ or $f_{1}'(k_{p}^{*};\, a)=0$. If $f_{1}(k_{p}^{*};\, a)=0$
then $k_{p}^{*}\in\frac{\pi}{a}\mathbb{Z}$ and applying lemma \ref{lem:np_at_vertex}
gives $\frac{\pi}{2b+2c}\in\frac{\pi}{a}\mathbb{Z}$, which contradicts
the incommensurability assumption. Otherwise, if $f_{1}'(k_{p}^{*};\, a)=0$,
we similarly obtain $\frac{\pi}{2b+2c}\in\frac{\pi}{2a}+\frac{\pi}{a}\mathbb{Z}$,
and again get a contradiction.

\end{proof}

\begin{lem}\label{lem:d_p_formula}

Let the set $\left\{ k_{p}^{*}\right\} $, as defined in lemma \ref{lem:np_at_vertex},
be the set of $k$ values at which merge and split events occur, and
let $k_{0}^{*}=0$. Denote $d_{p}:=\left|\left\{ k^{2}\in\sigma\left(\Gamma\right)\left|k_{p-1}^{\ast}<k<k_{p}^{\ast}\right.\right\} \right|$
, the number of eigenvalues of $\Gamma$ that occurred between two
consequent merge/split events. Then \begin{equation}
d_{p}=\left\lfloor \frac{a}{b+c}p\right\rfloor -\left\lfloor \frac{a}{b+c}\left(p-1\right)\right\rfloor +1.\label{eq:d_p_explicit_formula}\end{equation}

\end{lem}

\begin{proof}

The following two observations concern the set $D_{1}\left(k\right)\cup D_{2}\left(k\right)$,
which gives the positions of the nodal points on the leads.

The spectrum of $\Gamma$ may be characterized as \begin{equation}
\sigma\left(\Gamma\right)=\left\{ k^{2}\,\left|\, a\in D_{1}\left(k\right)\cup D{}_{2}\left(k\right)\right.\right\} .\label{eq:spectrum_char}\end{equation}

The merge/split events happen at \begin{equation}
\left\{ k_{p}^{*}\right\} =\left\{ k\,\left|\,0\in D_{1}\left(k\right)\cup D_{2}\left(k\right)\right.\right\} .\label{eq:split-merge_char}\end{equation}

We denote $Q\left(k\right):=\left(\left.0,a\right.\right]\cap\left\{ D_{1}\left(k\right)\cup D_{2}\left(k\right)\right\} $
and describe how it changes with $k$. Lemma \ref{lem:np_move_towards_cycle}
implies that the values of $Q\left(k\right)$ continuously decrease
with $k$. In addition, the first observation gives that $\left|Q\left(k\right)\right|$
increases at $k\in\sigma\left(\Gamma\right)$, when a nodal point
enters $\Gamma$. The second observation shows that $\left|Q\left(k\right)\right|$
decreases by one at $k\in\left\{ k_{p}^{*}\right\} $ , when a nodal
point enters the cycle. It is therefore evident that during the interval
$\left(\left.k_{p-1}^{*},k_{p}^{*}\right.\right],$ $\left|Q\left(k\right)\right|$
decreased a single time (at $k_{p}^{*}$), and the number of times
it increased is given by $d_{p}$, the number of eigenvalues in this
interval. We conclude that \begin{equation}
d_{p}=\left|Q\left(k_{p}^{*}\right)\right|-\left|Q\left(k_{p-1}^{*}\right)\right|+1.\label{eq:d_p_implicit_formula}\end{equation}
 It is easy to see that $D_{1}\left(k_{p}^{*}\right)\cup D_{2}\left(k_{p}^{*}\right)=\frac{\pi}{2k_{p}^{*}}\mathbb{Z}$,
and therefore \[
\left|Q(k_{p}^{*})\right|=\left|\left(\left.0,a\right.\right]\cap\frac{\pi}{2k_{p}^{*}}\mathbb{Z}\right|=\left\lfloor \frac{2k_{p}^{\ast}a}{\pi}\right\rfloor .\]
 Substituting $k_{p}^{*}=\frac{\pi}{2b+2c}p$ (lemma \ref{lem:np_at_vertex})
and plugging this in \eqref{eq:d_p_implicit_formula} gives \eqref{eq:d_p_explicit_formula}.

\end{proof}

We now have all the required information to obtain an expression for
$\mu_{n}$, the number of nodal points on $\Gamma$ at $k=k_{n}$.

\begin{proof} {[}Proof of theorem \ref{thm:dihedral_formula}{]}

In order to prove \eqref{eq:dihedral_nodal_points_formula} we need
to keep track of all the events which affect the number of nodal points
on the graph. These include the eigenvalues of the original graph,$\left\{ k_{n}\right\} $,
and the merge/split events, $\left\{ k_{p}^{*}\right\} $. Figure
\ref{fig:stream_of_events} shows a possible scenario for such a stream
of events. In this figure, the value of $\mu_{n}-n$ is shown for
each eigenvalue. The bounds on $\mu_{n}-n$ can be obtained from \eqref{eq:nodal_zeros_bound}
with a slight modification due to the additional nodal point positioned
on the Dirichlet vertex of the graph: 0$\leq\mu_{n}-n\leq1$. The
value of $\mu_{n}-n$ differs from $\mu_{n-1}-\left(n-1\right)$ if
and only if a merge/split event occurred in between the corresponding
eigenvalues. We therefore conclude that the value of $\mu_{n}-n$
depends on the parity of the number of merge/split events that occurred
before $k_{n}$.

\begin{figure}[h]
 \includegraphics[scale=0.55]{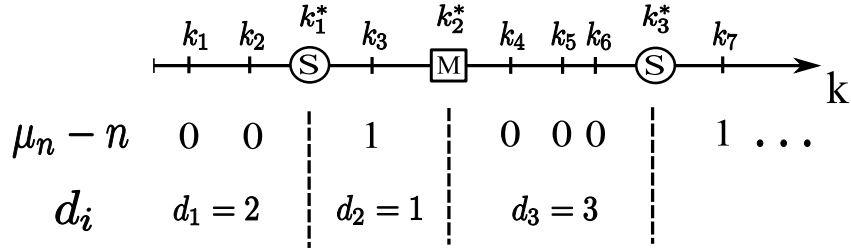}

\caption{An example of a stream of events needed to keep track of the number
of nodal points.\label{fig:stream_of_events}}

\end{figure}

Namely,

\[
\mu_{n}-n=\mmod_{2}p,\]
 where $p$ is an integer such that \[
k_{p}^{\ast}<k_{n}<k_{p+1}^{\ast}.\]

By the definition of $d_{p}$ (see lemma \ref{lem:d_p_formula}) this
is equivalent to\[
\sum_{i=1}^{p}d_{i}<n\leq\sum_{i=1}^{p+1}d_{i},\]

which by \eqref{eq:d_p_explicit_formula} evaluates to\[
\left\lfloor \frac{a}{b+c}p\right\rfloor +p<n\leq\left\lfloor \frac{a}{b+c}\left(p+1\right)\right\rfloor +\left(p+1\right).\]

Since $n,\, p$ are integers and $\frac{a}{b+c}\notin\mathbb{Q}$,\[
\frac{a}{b+c}p+p<n\leq\frac{a}{b+c}\left(p+1\right)+\left(p+1\right).\]

Multiplying through by $\frac{b+c}{a+b+c}$ we get

\[
p<\frac{b+c}{a+b+c}n<\left(p+1\right),\]

and conclude that \[
p=\left\lfloor \frac{b+c}{a+b+c}n\right\rfloor .\]

The number of nodal points on the graph is therefore given by \[
\mu_{n}=n+\mmod_{2}\left(\left\lfloor \frac{b+c}{a+b+c}n\right\rfloor \right).\]

We now wish to turn this into a formula for the nodal count, $\nu_{n}$.
The relation between $\mu_{n}$ and $\nu_{n}$ depends on whether
the $n^{th}$ eigenfunction has nodal points on the loop as demonstrated
in figure \ref{fig:points_and_domains}.

\begin{figure}[h]
 \hfill{}\includegraphics[scale=0.4]{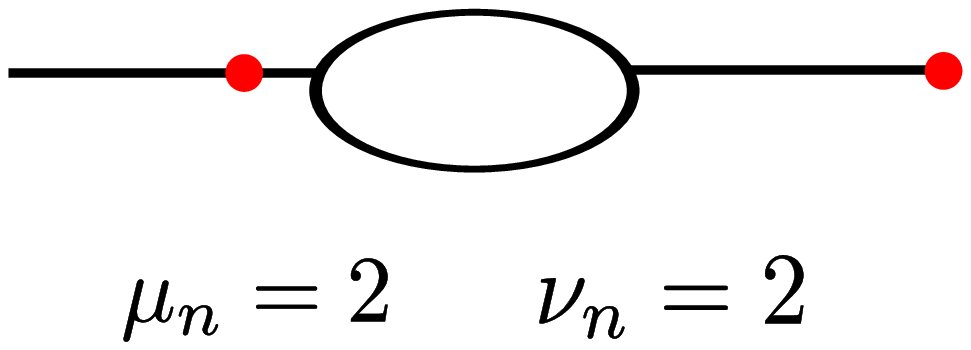}\hfill{}\includegraphics[scale=0.4]{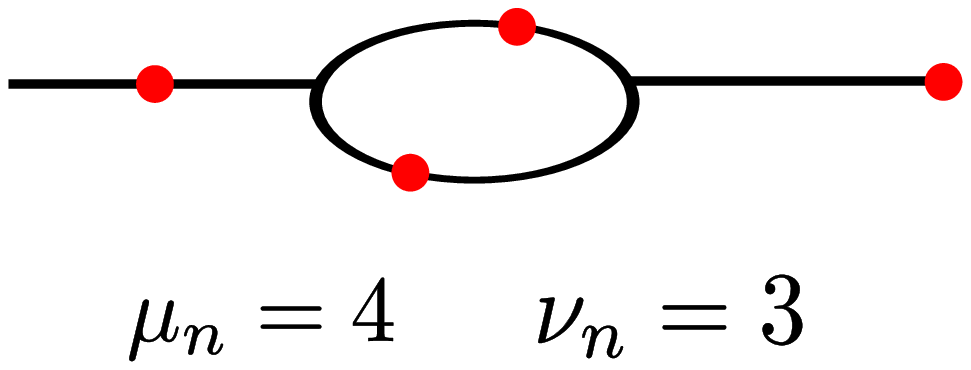}\hfill{}

\bigskip{}

\hfill{}(a)\qquad{}\hfill{}\qquad{}(b)\hfill{}

\caption{The two possible relations between the number of nodal points, $\mu_{n}$,
and the number of nodal domains, $\nu_{n}$.\label{fig:points_and_domains}}

\end{figure}

If it does have nodal points on the loop then $\nu_{n}=\mu_{n}-1$
(figure \ref{fig:points_and_domains}(b)), and in the case it does
not, $\nu_{n}=\mu_{n}$ (figure \ref{fig:points_and_domains}(a)).\\
 We therefore have that for the first $d_{1}=\left\lfloor \frac{a}{b+c}\right\rfloor +1$
eigenvalues (when there are still no nodal points on the loop) the
nodal count is \[
\nu_{n}=n+\mmod_{2}\left(\left\lfloor \frac{b+c}{a+b+c}n\right\rfloor \right)=n,\]
 where the second equality is due to $n\leq\left\lfloor \frac{a}{b+c}\right\rfloor +1$.
For the rest of the nodal count, $n>\left\lfloor \frac{a}{b+c}\right\rfloor +1$,
we get \[
\nu_{n}=n-1+\mmod_{2}\left(\left\lfloor \frac{b+c}{a+b+c}n\right\rfloor \right).\]

\end{proof}



\section{The nodal count of graphs with disjoint cycles}
\label{sec:disjoint_cycles}

\subsection{Graphs with $\beta=1$: a dynamical approach}
\label{sec:single_cycle_dyn}

In this section we will discuss the nodal dynamics on a graph with one
cycle (i.e. $\beta=1$) and a lead attached to a general position on
the cycle, see figure~\ref{fig:one_loop}.  The discussion will not be
formal, as we will prove the results by other methods in
section~\ref{sec:single_cycle_proof}.

\begin{figure}[th]
  \includegraphics[scale=0.6]{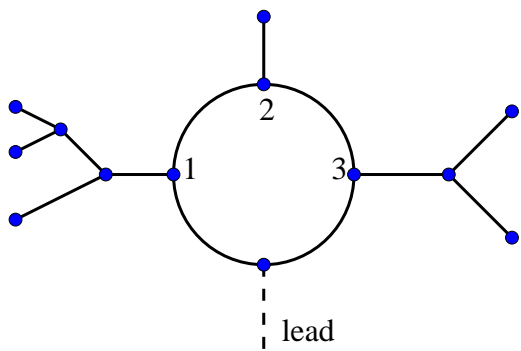}
  \caption{An example of a graph $\Gamma$ with $\beta=1$.  We view it
    as a cycle with trees grafted upon it.}
  \label{fig:one_loop}
\end{figure}

We have seen in section \ref{sec:single_lead_dyn} that, as $k$
increases the nodal points (zeroes) travel along the lead in the
direction of the graph.  Consider the quantity $\delta_n = \mu_n - n + 1$.
This is the ``surplus'' of zeros due to the graph not being a tree.
Bound \eqref{eq:nodal_zeros_bound} implies that $\delta_n$ can be equal to
either $0$ or $1$.  The change in this quantity from eigenvalue
$k=k_{n-1}$ to eigenvalue $k=k_{n}$ can be attributed to the following
three causes:
\begin{enumerate}
\item The increase in the index $n$ (the change in $d$ is $-1$).
\item A zero entering the graph from the lead.  Upon entering, the
  zero can either merge (M) with a zero already present on the cycle
  or split (S) into two zeros.
\item A zero entering a tree. This zero can either split off a
  zero traveling on the cycle or it can be a result of two zeros from
  the cycle merging together.
\end{enumerate}
Another notable event is a zero {\em passing\/} through a vertex where
a tree is attached.  We did not list it above since an event of this
type does not affect the nodal count. Similarly, when a zero is
traveling through the tree, we know (see
\cite{PokPryObe_mz96,Sch_wrcm06,Ber_cmp08}) that the number of zeros
does not change.

As already explained in section~\ref{sec:single_lead_dyn}, event (2)
happens exactly once between each pair of eigenvalues $k_{n-1}$ and
$k_{n}$, since an eigenvalue corresponds to the Neumann condition and
the entrance event corresponds to the Dirichlet condition satisfied at
the attachment point.  If event (2) is a split, the contribution to
$\delta$ is $+2$, otherwise it is $0$.  However, if we consider the
total contribution of events (1) and (2), we get $+1$ from a split and
$-1$ from a merge.  This is the same as a contribution of a type (3)
event, when the split results in $+1$ (number of zeros on the cycle
stays the same but another zero appears on a tree) and the merge in
$-1$ (the number of zeros on the cycle reduces by $2$, while one zero
enters a tree).

The first eigenfunction has constant sign, so $\delta_1=0$ and no
events happen until $k=k_1$.  Since the contribution of type (1) is
now absorbed in the contributions of type (2), the value of $\delta_n$
is the total number of splits minus the total number of merges up to
$k=k_n$.  On the other hand, $\delta_n$ is restricted by the nodal
bound to be either $1$ or $0$, therefore it is equal to the
\emph{parity of the total number of S/M events}.

There are exactly $n-1$ events of type (2) happening until $k=k_n$.
To count the number of events of type (3), we consider an auxiliary
graph $\Gamma_\emptyset$, obtained from $\Gamma$ by removing all edges
belonging to the cycle and imposing Dirichlet conditions on the points
where the trees were connected to the cycle.  The graph
$\Gamma_\emptyset$ is a collection of trees that were grafted on the
cycle.  Since a zero entering a tree signals that the Dirichlet
condition is satisfied on the tree, the corresponding value of $k$ is
in fact an eigenvalue\footnote{The corresponding eigenfunction is
  identically zero on all trees apart from the one with Dirichlet
  condition satisfied.} of $\Gamma_\emptyset$.  And the number of
events of type (3) is thus equal to the number of eigenvalues of
$\Gamma_\emptyset$ that are smaller than $k_n$.  To summarize,
\begin{equation*}
  \delta_n = \mmod_2(n-1+N_\emptyset(k_n)),
\end{equation*}
where $N_\emptyset(k_n)$ is the spectral counting function of the
graph $\Gamma_\emptyset$.  Thus we can fully predict the nodal count
using the spectra of two graphs, $\Gamma$ and $\Gamma_\emptyset$.  The
discussion above captures the dynamics of the zeros, but it is
relatively difficult to formalize.  Instead we will prove the formula
for $\mu_n$ by other methods, which, although not very pictorial,
allow us to extend the argument to the case of non-zero potential
$V(x)$.

\subsection{Graphs with $\beta=1$: a formal proof}
\label{sec:single_cycle_proof}

In this section we prove the formula that was informally
derived in section~\ref{sec:single_cycle_dyn}.

\begin{thm}
  \label{thm:single_cycle_formula}
  Consider the Schr\"odinger operator \eqref{E:schrod} on a connected
  graph with a single cycle.  Let the $n$-th eigenvalue be simple and
  the corresponding eigenfunction be non-zero on the vertices. Then
  \begin{equation}
    \mu_{n}=n-1+\mmod_{2}\left(n-1+N_{\emptyset}(\lambda_{n})\right),
    \label{eq:formula_one_cycle}
  \end{equation}
  where $N_{\emptyset}$ is the spectral counting function of the
  disconnected graph obtained by removing the cycle and putting
  Dirichlet conditions on the new vertices.
\end{thm}

\begin{proof}
  From the nodal bound for graphs with one cycle (i.e.  with
  $\beta=1$) we know that $\mu_{n}$ is equal to $n-1$ or $n$.  The
  first step of the proof is to observe that the number of zeros on
  the edges that do not belong to the cycle is equal to
  $N_{\emptyset}(\lambda_{n})$.  We will prove this statement
  below. Then $\mu_{n}-N_{\emptyset}(\lambda_{n})$ is the number of
  zeros on the cycle, and has to be even, as explained in
  section~\ref{sec:intro_nodal_count}.

  First, assume that $\mu_{n}=n$. Then, the quantity
  \begin{equation*}
  \mu_{n}-1-N_{\emptyset}(\lambda_{n})=n-1-N_{\emptyset}(\lambda_{n})
  \end{equation*}
  is odd and therefore
  \begin{equation*}
  \mmod_{2}\left(n-1+N_{\emptyset}(\lambda_{n})\right)
  = \mmod_{2}\left(n-1-N_{\emptyset}(\lambda_{n})\right)=1,
  \end{equation*}
  where we used the fact that
  \begin{equation*}
    \mmod_{2}(a+b)=\mmod_{2}(a-b)
  \end{equation*}
  for any integer $a$ and $b$.  Thus the right-hand side of
  equation~(\ref{eq:formula_one_cycle}) evaluates to $n$ which is the
  right answer.

  If $\mu_{n}$ is not equal to $n$, then it is equal to $n-1$ and we
  have
  \begin{equation*}
  \mmod_{2}\left(n-1+N_{\emptyset}(\lambda_{n})\right)
  = \mmod_{2}\left(\mu_{n}-N_{\emptyset}(\lambda_{n})\right)=0,
  \end{equation*}
  since $\mu_{n}-N_{\emptyset}(\lambda_{n})$ is even.  Thus
  equation~(\ref{eq:formula_one_cycle}) still holds.

  Now we prove that $N_{\emptyset}(\lambda_{n})$ is indeed the number
  of zeros on the subtrees of the graph. To shorten the formulas we
  introduce the following notation. We denote the $n$-th eigenvalue by
  $\Lambda$ and the corresponding eigenfunction by $F$. Break up the
  original graph into the cycle and the trees $T_{j}$. For each tree
  $T_{j}$ we choose as a root the vertex that was its contact point
  with the cycle. We can ensure that each root has degree 1: if
  necessary we can split trees that share a root. On each tree the
  vertex conditions are inherited from the graph, but we still need to
  specify the conditions on the root. We will consider two versions of
  each tree. The first, $T_{j,F}$ has the condition on the root $r$
  chosen to be satisfied by the function $F$, restricted to the
  tree. That is we chose the constant $\alpha$ in the $\delta$-type
  condition to be $\alpha=F'(r)/F(r)$. The second version of the tree,
  denoted $T_{j,\infty}$, has the Dirichlet condition on the root.

  Denote by $T_{\infty}$, the disjoint union of the graphs
  $T_{j,\infty}$. We observe that
  \begin{equation*}
    N_{\emptyset}(\lambda)=N_{T_{\infty}}(\lambda)
    =\sum_{j}N_{T_{j,\infty}}(\lambda).
  \end{equation*}
  Thus we only need to prove that $N_{T_{j,\infty}}(\Lambda)$ gives
  the number of zeros of $F$ on the subtree $T_{j,\infty}$. Since, by
  construction, the restriction of $F$ is an eigenfunction of
  $T_{j,F}$ with the eigenvalue $\Lambda$, we have
  $\Lambda=\lambda_{m}(T_{j,F})$ for some $m$. By the strict
  interlacing, theorem~\ref{thm:interlacing_monotone},
  \begin{equation*}
    \lambda_{m-1}(T_{j,\infty})<\lambda_{m}(T_{j,F})<\lambda_{m}(T_{j,\infty}),
  \end{equation*}
  and, therefore, $N_{T_{j,\infty}}(\Lambda)=m-1$. On the other hand,
  the nodal count on trees, equation~\eqref{eq:nodal_zeros_bound},
  gives $\mu_{m}(T_{j,\infty})=m-1$.  This concludes the proof.
\end{proof}

\subsection{Number of zeros on a graph with disjoint cycles}

In fact, the formula of the previous section can be extended to
$\beta>1$ as long as the cycles do not share any vertices.

\begin{thm}
  \label{thm:disjoint_cycle_formula}
  For a connected graph containing $\beta$ disjoint cycles, let the
  $n$-th eigenvalue be simple and the corresponding eigenfunction be
  non-zero on the vertices. Then
  \begin{equation*}
    \mu_{n} = n-1
    + \sum_{j=1}^\beta \mmod_2\left(n-1+N_{j,\emptyset}(\lambda_n)\right),
  \end{equation*}
  where $N_{j,\emptyset}$ is the spectral counting function of the
  disconnected graph obtained by removing the $j$-th cycle and putting
  Dirichlet conditions on the new vertices.
\end{thm}

\begin{proof}
  Denote the $n$-th eigenvalue by $\Lambda$ and the corresponding
  eigenfunction by $F$.  Choose an arbitrary cycle and let $e_{1}$,
  \ldots{}, $e_{k}$ be the edges incident to it.  Since the cycles are
  disjoint, these edges do not belong to any cycle.  Choose points
  $c_{1}$, \ldots{}, $c_{k}$, one on each edge, so that the function
  $F$ is non-zero at these points. If the graph is cut at these
  points, we obtain $k+1$ disjoint subgraphs, $\Gamma_{j}$,
  $j=0,\ldots,k$ (the $0$-th subgraph contains the chosen cycle), see
  figure~\ref{fig:dc1}.

  \begin{figure}[th]
    \includegraphics[scale=1]{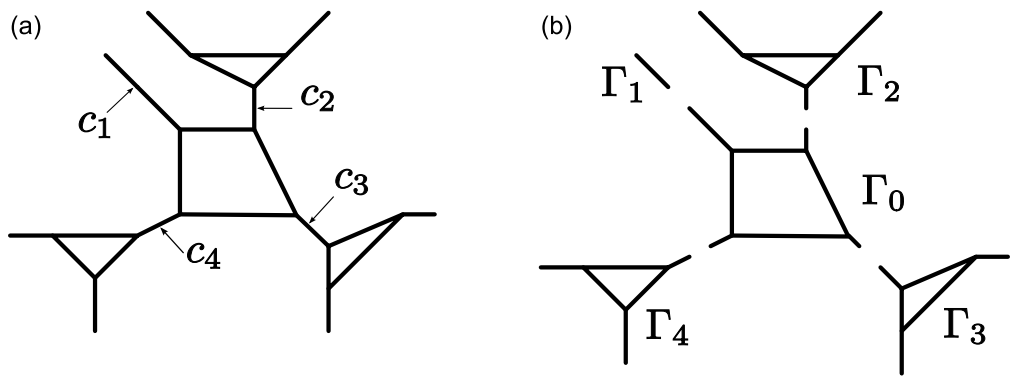}
    \caption{An example of a graph with 4 disjoint cycles.  In part
      (a), the middle cycle is chosen and the cut-points are labeled.
      In part (b) the graph is split up into subgraphs.  The central
      part $\Gamma_0$ contains exactly one cycle, while the other parts
      can contain more or fewer.}
    \label{fig:dc1}
  \end{figure}

  Define
  \begin{equation*}
    \alpha_{c_{j}}=\frac{F'(c_{j})}{F(c_{j})},
  \end{equation*}
  where the derivative is taken away from the chosen loop. We impose
  $\delta$-type conditions on the newly formed vertices. The vertex
  $c_{j}$ belonging to $\Gamma_{j}$ will get the condition with
  coefficient $\alpha_{c_{j}}$ and its counterpart belonging to the
  subgraph $\Gamma_{0}$ will get the condition with coefficient
  $-\alpha_{c_{j}}$. This way, the appropriately cut function $F$ is
  still an eigenfunction on all $k+1$ subgraphs and $\Lambda$ is the
  corresponding eigenvalue.  This allows us to define $m_{j}$ by
  \begin{equation*}
    \Lambda=\lambda_{m_{j}}(\Gamma_{j}).
  \end{equation*}

  \begin{lem}
    \label{lem:sum_of_m}
    The numbers $m_{j}$ are well-defined and satisfy
    \begin{equation}
      n-1=\sum_{j=0}^{k}(m_{j}-1).
      \label{eq:sum_of_m}
    \end{equation}
  \end{lem}

  \begin{proof}[Proof of the lemma]
    Let $\Gamma_{c}$ denote the disjoint union of the graphs
    $\Gamma_{j}$, $j=0,\ldots,k$.  First of all, we apply
    theorem~\ref{thm:interlacing_join} $k$ times (for $k$ cuts) to
    obtain inequalities
    \begin{equation*}
      \lambda_{n-1}(\Gamma_{c}) \leq\lambda_{n-1}(\Gamma)
      \qquad
      \mbox{and}
      \qquad
      \lambda_{n+1}(\Gamma) \leq \lambda_{n+k+1}(\Gamma_{c})
    \end{equation*}
    On the other hand, simplicity of the eigenvalue
    $\lambda_n(\Gamma)=\Lambda$ means that $\lambda_{n-1}(\Gamma) <
    \Lambda < \lambda_{n+1}(\Gamma)$ and therefore
    \begin{equation}
      \label{eq:Lambda_and_Gamma_c}
      \lambda_{n-1}(\Gamma_{c}) < \Lambda < \lambda_{n+k+1}(\Gamma_{c}).
    \end{equation}
    Finally, out of $F$ we can form at least $k+1$ linearly
    independent eigenfunctions of the graph $\Gamma_c$: functions that
    are restrictions of $F$ on one of the parts $\Gamma_j$ and
    identically zero on all the others.  All these eigenfunctions have
    eigenvalue $\Lambda$.  Combining this observation with
    inequality~\eqref{eq:Lambda_and_Gamma_c} we conclude that
    $\Lambda$ has degeneracy exactly $k+1$ in the spectrum of
    $\Gamma_c$ and therefore is a simple eigenvalue of every part
    $\Gamma_j$.  Thus the numbers $m_{j}$ are well-defined.

    Finally, since the spectrum of $\Gamma_{c}$ is the superposition
    of spectra of $\Gamma_{j}$, equation \eqref{eq:sum_of_m} is
    equivalent to the statement {}``there are $n-1$ eigenvalues of
    $\Gamma_{c}$ that are strictly less than $\Lambda$'', which is
    also obvious from inequality~\eqref{eq:Lambda_and_Gamma_c} and the
    fact that $\lambda_n(\Gamma_c)=\Lambda$.
  \end{proof}

  \begin{figure}[th]
    \includegraphics[scale=1]{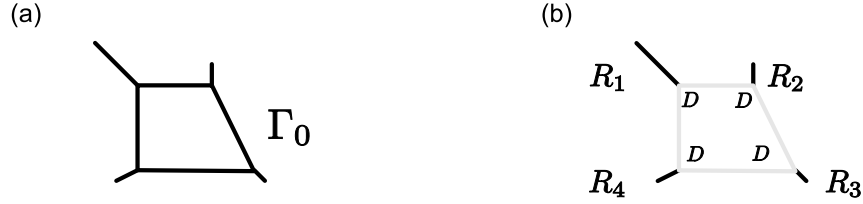}
    \caption{The subgraph $\Gamma_0$ and the graph $R$ obtained after
      removing the cycle (it appears shaded on part (b) of the
      figure).  The graph $R$ is a disjoint union of four subgraphs
      $R_j$.}
    \label{fig:dc2}
  \end{figure}

  We now want to use theorem~\ref{thm:single_cycle_formula} to find
  the number of zeros of the function $F$ on the graph $\Gamma_{0}$.
  Let $R$ be the graph obtained from $\Gamma_{0}$ by removing the
  cycle and imposing Dirichlet conditions on the new vertices. This
  graph is a disjoint union of the graphs $R_{j}$, $j=1,\ldots,k$, see
  Fig.~\ref{fig:dc2} Therefore, we have
  \begin{equation}
    N_{R}(\Lambda)=\sum_{j=1}^{k}N_{R_{j}}(\Lambda).
    \label{eq:R_spec}
  \end{equation}
  According to theorem~\ref{thm:single_cycle_formula} the number of
  zeros of $F$ on the subgraph $\Gamma_{0}$ is
  \begin{equation*}
    \mu(\Lambda,\Gamma_{0})
    = m_{0}-1+\mmod_{2}\left(m_{0}-1+N_{R}(\Lambda)\right).
  \end{equation*}
  Extracting $m_{0}-1$ from equation~\eqref{eq:sum_of_m} and using
  equation~\eqref{eq:R_spec} we get
  \begin{equation}
    \mu(\Lambda,\Gamma_{0})
    = m_{0}-1+\mmod_{2}\left(n-1
      +\sum_{j=1}^{k}\left(m_{j}-1+N_{R_{j}}(\Lambda)\right)\right),
    \label{eq:mu_gamma0_1}
  \end{equation}
  where we used
  \begin{equation*}
    \mmod_{2}(a+b)=\mmod_{2}(a-b)
  \end{equation*}
  for integer $a$ and $b$ to change some signs.
  Define now the graph $G$ by removing the chosen cycle from the
  original graph $\Gamma$ and imposing the Dirichlet conditions on the
  new vertices.  Similarly to the graph $R$, the graph $G$ is a
  disjoint union of $k$ subgraphs $G_{j}$, see Fig.~\ref{fig:dc3}, and
  \begin{equation}
    N_{G}(\Lambda)=\sum_{j=1}^{k}N_{G_{j}}(\Lambda).
    \label{eq:G_spec}
  \end{equation}
  If we were to cut the graph $G_{j}$ at the point $c_{j}$ the two
  parts would be exactly $\Gamma_{j}$ and $R_{j}$. This suggests the
  following lemma.

  \begin{figure}[th]
    \includegraphics[scale=1]{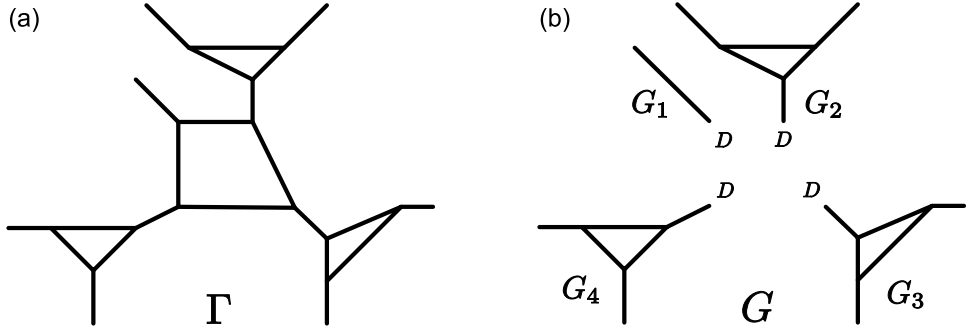}
    \caption{The original graph $\Gamma$ and the graph $G$ obtained by
      removing the chosen cycle.  The graph $G$ is a disjoint union of
      graphs $G_j$.}
    \label{fig:dc3}
  \end{figure}

  \begin{lem}
    For every $j=1,\ldots,k$,
    \begin{equation}
      N_{G_{j}}(\Lambda)=m_{j}-1+N_{R_{j}}(\Lambda).
      \label{eq:gamma_plus_R}
    \end{equation}
  \end{lem}

  \begin{proof}[Proof of the lemma]
    First we observe that $\Lambda$ belongs to the spectrum of the
    graph $\Gamma_{j}$ and does not belong to the spectrum of $R_{j}$
    or $G_{j}$. This can be shown by the strict interlacing,
    theorem~\ref{thm:interlacing_monotone}, applied to the graph $\Gamma$
    (corresp. $G_0$) by changing the condition from Neumann to
    Dirichlet at the vertex where $G_j$ (corresp. $R_j$) was connected
    to the cycle.

    Denote by $G_{j,c}$ the disjoint union of the graphs $\Gamma_{j}$
    and $R_{j}$. Let integer $q$ be such that
    \begin{equation*}
      \Lambda=\lambda_{q}(G_{j,c}).
    \end{equation*}
    Since $\Lambda=\lambda_{m_{j}}(\Gamma_{j})$, we have that
    $q=m_{j}+N_{R_{j}}(\Lambda)$.  On the other hand, by
    theorem \ref{thm:interlacing_join},
    \begin{equation*}
      \lambda_{q-1}(G_{j})<\Lambda=\lambda_{q}(G_{j,c})<\lambda_{q}(G_{j})\end{equation*}
    Therefore, $N_{G_{j}}(\Lambda)=q-1$ which concludes the
    proof.
  \end{proof}

  Combining equations \eqref{eq:G_spec} and \eqref{eq:gamma_plus_R}
  with formula \eqref{eq:mu_gamma0_1} we arrive to
  \begin{equation}
    \mu(\Lambda,\Gamma_{0})
    =m_{0}-1+\mmod_{2}\left(n-1+N_{G}(\Lambda)\right).
    \label{eq:mu_gamma0_2}
  \end{equation}
  We should also note that $N_{G}=N_{j,\emptyset}$, where $j$ is the
  number of the cycle that was chosen.

  Since we chose an arbitrary cycle, equation \eqref{eq:mu_gamma0_2}
  is valid for every cycle. The conclusion of the theorem is just the
  sum of equations \eqref{eq:mu_gamma0_2} over all cycles with an
  application of the analogue of equation \eqref{eq:sum_of_m}.
\end{proof}

\section{A discussion }

\subsection{An approximate derivation of an exact nodal count formula}

We will now present an alternative way to get the nodal points count
formula \eqref{eq:dihedral_nodal_points_formula} of the graph given in
figure \ref{fig:original_dihedral_without_cycle}(a). The derivation is
most appealing, but involves an approximation that cannot be
justified.  We present it here because it makes use of an idea which
has been used in other contexts.  Also we find that an unjustifiable
approximation that reproduces the exact nodal count formula carries
information about the graph in its own right.

\begin{figure}[ht]
\includegraphics[scale=0.55]{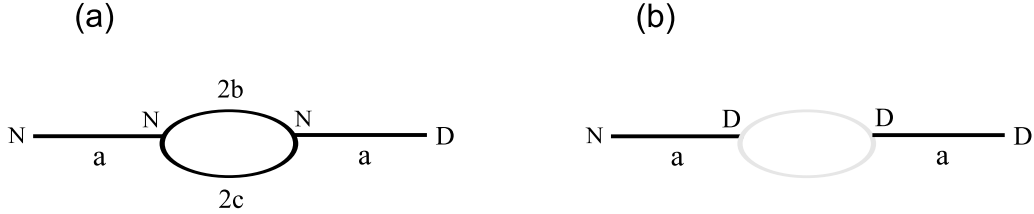}

\caption{\label{fig:original_dihedral_without_cycle}(a) The quantum graph
whose nodal count we compute. (b) The same graph with its cycle removed
(appears shaded).}

\end{figure}

We start by rewriting the formula \eqref{eq:formula_one_cycle} with
a slight modification, due to the nodal point which is positioned
on the Dirichlet boundary vertex of the graph:

\[
\mu_{n}=n+\mmod_{2}\left(n+N_{\emptyset}(k_{n})\right).\]
The spectra of the two edges which appear in figure \ref{fig:original_dihedral_without_cycle}(b)
are $\left\{ \frac{\pi}{a}n\right\} _{n\in\mathbb{N}}$ and  $\left\{ \frac{\pi}{a}\left(n+\frac{1}{2}\right)\right\} _{n\in\mathbb{N}}$.
Their spectral counting function is therefore\[
N_{\emptyset}(k)=\left\lfloor \frac{ka}{\pi}\right\rfloor +\left\lfloor \frac{ka}{\pi}+\frac{1}{2}\right\rfloor .\]
Plugging it above and using the identities $\mmod_{2}\left(\left\lfloor x\right\rfloor +\left\lfloor x+\frac{1}{2}\right\rfloor \right)=\mmod_{2}\left(\left\lfloor 2x\right\rfloor \right)$
and $\mmod_{2}\left(x+y\right)=\mmod_{2}\left(x-y\right)$ we obtain

\[
\mu_{n}=n+\mmod_{2}\left(n-\left\lfloor \frac{2k_{n}a}{\pi}\right\rfloor \right).\]
We can get an approximate expression for $k_{n}$ from the Weyl term
of the spectral counting function of the graph, \[
N\left(k\right)\approx\frac{2(a+b+c)}{\pi}k,\]
by its inversion, i.e.,

\[
k_{n}\approx\frac{\pi}{2(a+b+c)}n.\]
Using this approximation gives\begin{align*}
\mu_{n} & \approx n+\mmod_{2}\left(n-\left\lfloor \frac{a}{a+b+c}n\right\rfloor \right)\\
 & =n+\mmod_{2}\left(\left\lfloor \frac{b+c}{a+b+c}n\right\rfloor \right),\end{align*}
which is the exact result, \eqref{eq:dihedral_nodal_points_formula}.

One should note that the last step of the derivation, which involves
an approximation of $k_{n}$ by inverting the Weyl term, cannot be justified.
Moreover, the floor function is a discontinuous function and it is
therefore expected that an approximation of its argument would lead
to a completely wrong result for some portion of the sequence.

From the exactness of the final result, we conclude the following
property of the spectrum\[
N_{\emptyset}(k_{n})=\mmod_{2}\left(\left\lfloor \frac{2a}{\pi}\, k_{n}\right\rfloor \right)=\mmod_{2}\left(\left\lfloor \frac{2a}{\pi}\,\frac{\pi}{2(a+b+c)}n\right\rfloor \right).\]
Numeric examination reveals that the equality hold for the arguments of the $\mmod_{2}$ as well, namely\begin{equation}
\left\lfloor \frac{2a}{\pi}\, k_{n}\right\rfloor =\left\lfloor \frac{2a}{\pi}\,\frac{\pi}{2(a+b+c)}n\right\rfloor .\label{eq:surprising_spectral_identity}\end{equation}
The above relation connects the spectrum and the lengths of the graph's
edges. Having such a relation for our graph makes $N_{\emptyset}(k_{n})$
expressible in terms of the parameters $a,b,c$ and enables to turn
the nodal count formula, \eqref{eq:formula_one_cycle}, into a formula
which contains geometric properties of the graph, rather than spectral
ones. In short, the special nodal count formula is a direct consequence
of a purely spectral identity - a connection between the graph's spectrum
and the spectral counting function of its subgraphs.

The novelty of this result makes one wonder to what extent it can
be generalized to other graphs. Even if such an exact result is not
reproduced, one may still use approximations of the type above and
try to estimate the errors caused by them.

\subsection{Periodic orbits expansions}

Wishing to express the nodal count formula \eqref{eq:dihedral_nodal_points_formula}
as a periodic orbits expansion, we notice that $\mmod_{2}\left(\left\lfloor x\right\rfloor \right)$
is an odd periodic function (of period 2), whose Fourier transform
is: \[
\mmod_{2}\left(\left\lfloor x\right\rfloor \right)=\frac{1}{2}-\sum_{k=0}^{\infty}\frac{2}{(2k+1)\pi}\sin\left(\left(2k+1\right)\pi x\right).\]

Denoting $\alpha:=\frac{b+c}{a+b+c}$, the normalized length of the
loop, we can rewrite \eqref{eq:dihedral_nodal_points_formula} as
following: \[
\mu_{n}=n+\frac{1}{2}-\sum_{k=0}^{\infty}\frac{2}{(2k+1)\pi}\sin\left(\left(2k+1\right)\pi\alpha n\right).\]\label{eq:periodic_orbits_of_dihedral}

We therefore get that the nodal points sequence is expressed in terms
of lengths of periodic orbits on the graph. One should note that the
only periodic orbits that appear are odd repetitions of the graph's
cycle. They appear with harmonically decaying amplitudes. This causes
to seek for a more direct derivation of the periodic orbits expansion
which may also explain the meaning of the amplitudes and the absence
of other periodic orbits. Furthermore, the formula \eqref{eq:formula_one_cycle},
which holds for any graph with a single cycle, may also be turned
into an expansion of a similar type. We recall that for quantum graphs
there exist an exact periodic expansion for the spectral counting
function. Therefore, the spectral counting function of the subgraph,
$N_{\emptyset}(k_{n})$, can be expanded and plugged in formula \eqref{eq:formula_one_cycle}.
This would yield an expansion which still involves the spectral information,
$\left\{ k_{n}\right\} $. Having an approximate inversion of the
spectral counting function of the whole graph then enables to further
get a periodic orbits formula which involves only geometric properties
of the graph. Such spectral inversion attempts were recently carried out with a high degree of success \cite{GKS06, PhysRevLett.88.044101}. It is therefore evident that the obtained result leads to a wide field of further questions and open research possibilities.

\section{acknowledgments}
It is a pleasure to acknowledge Sven Gnutzmann for fruitful
discussions about the scattering matrix properties. We are grateful to
Peter Kuchment for suggesting to extend
Theorem~\ref{thm:single_cycle_formula} to what is now
Theorem~\ref{thm:disjoint_cycle_formula}.  We also wish to thank Amit
Godel for the careful examination of the proof of theorem
\ref{thm:dihedral_formula}.  The work was supported by the Minerva
Center for Nonlinear Physics, the Einstein (Minerva) Center at the
Weizmann Institute and the Wales Institute of Mathematical and
Computational Sciences) (WIMCS). Grants from EPSRC (grant EP/G021287),
ISF (grant 166/09), BSF (grant 2006065) and NSF (DMS-0604859 and
DMS-0907968) are acknowledged.

\appendix
\section{Interlacing theorems for quantum graphs}

Eigenvalue interlacing (or bracketing) is a powerful tool in spectral
theory. In particular, in the graph setting, it allows to estimate
eigenvalue of a given graph via the eigenvalues of its subgraphs,
which may be easier to calculate.  Here we quote the theorems that are
used in the proofs of the formulas of the present manuscript.  The
theorems are quoted in the form they appear in \cite{BerKuc_prep10}.

The first theorem deals with choosing a vertex on the graph $\Gamma$
and changing the parameter $\alpha_v$ of the extended $\delta$-type
condition at $v$ (see equation~\eqref{eq:delta_deriv}).  We remind the
reader that $\alpha_v=\infty$ corresponds to the Dirichlet condition at the
vertex which essentially disconnects the edges meeting at the vertex.

\begin{thm}[Interlacing when changing a parameter]
  \label{thm:interlacing_monotone}
  Let $\Graph_{\alpha'}$ be the graph obtained from the graph
  $\Graph_{\alpha}$ by changing the coefficient of the condition at
  vertex $v$ from $\alpha$ to $\alpha'$.  If $-\infty < \alpha <
  \alpha' \leq \infty$, then
  \begin{equation}
    \label{eq:interlacing_monotone}
    \lambda_n(\Graph_\alpha) \leq \lambda_n(\Graph_{\alpha'}) \leq
    \lambda_{n+1}(\Graph_\alpha).
  \end{equation}
  If the $n$-th eigenvalue of $\Gamma_{\alpha'}$ is simple and the
  corresponding eigenfunction is nonzero on the vertices, the
  inequalities are strict.
\end{thm}

The second theorem deals with the situation when the graph $\Gamma'$
is obtained from $\Gamma$ by gluing two vertices together, or,
equivalently, $\Gamma$ is obtained by cutting the graph $\Gamma'$ at a
vertex or at a point on an edge.\footnote{Any point on an edge can be
  viewed as a vertex of degree $2$}  When gluing the vertices
together, their respective parameters $\alpha_{v_0}$ and
$\alpha_{v_1}$ get added.

\begin{thm}[Interlacing when gluing the vertices]
  \label{thm:interlacing_join}
  Let $\Graph$ be a compact (not necessarily connected) graph.  Let
  $v_0$ and $v_1$ be vertices of the graph $\Graph$ endowed with the
  $\delta$-type conditions with the parameters $\alpha_0$ and
  $\alpha_1$ (see definition \ref{D:spaces}).  Arbitrary self-adjoint
  conditions are allowed at all other vertices of $\Graph$.

  Let $\Graph'$ be the graph obtained from $\Graph$ by gluing the
  vertices $v_0$ and $v_1$ together into a single vertex $v$, so that
  $\Edges_v=\Edges_{v_0}\cup \Edges_{v_1}$, and endowed with the $\delta$-type
  condition with the parameter $\alpha_v = \alpha_0 + \alpha_1$.

  Then the eigenvalues of the two graphs satisfy the inequalities
  \begin{equation}
    \label{eq:interlacing_join}
    \lambda_n(\Graph) \leq \lambda_n(\Graph') \leq \lambda_{n+1}(\Graph).
  \end{equation}

  In addition, if $\lambda_n(\Graph')$ is simple and the corresponding
  eigenfunction is nonzero on vertices and not an eigenfunction of
  $\Gamma$, the inequalities are strict.
\end{thm}

An intuitive explanation for the above result is that by gluing
vertices we impose an additional restriction: the continuity
condition.  This additional restriction pushes the spectrum up.


\bibliographystyle{ieeetr}
\bibliography{bk_bibl,additional,Nodal_Domains_general,Nodal_Domains_of_Quantum_Graphs,Qunatum_Graphs,Trace_Formulae,Scattering}

\end{document}